\documentclass[oneside, a4paper]{amsart}
\usepackage[a4paper,  width=15cm, height=24cm]{geometry}
\newcommand{\truepar}{\bigskip \noindent}
\usepackage{subfiles}
\usepackage[]{currvita}
\usepackage[latin1]{inputenc}
\usepackage{amsfonts, amsmath, amssymb, amsthm}
\usepackage[T1]{fontenc}
\usepackage{lmodern,dsfont}
\usepackage[english]{babel}
\usepackage{fixltx2e,mparhack,url,varioref,}
\usepackage[stretch=10, shrink=10]{microtype} 
\usepackage{overpic}
\usepackage[auto]{contour}

\newcommand{\VV}{\mathcal{V}}
\newcommand{\PP}{\mathcal{H}}
\newcommand{\PL}{\mathbb{P}(\mathcal{P})}
\newcommand{\LL}{\mathbb{P}(\mathcal{L})}
\newcommand{\LLL}{\mathcal{L}}

\newcommand{\QQ}{\mathcal{Q}}
\newcommand{\s}{\mathfrak{s}}

\newcommand{\fa}{\mathfrak{a}}

\newcommand{\lspan}[1]{\langle{#1}\rangle}
\newcommand{\lspann}[1]{\text{span}\{ {#1}\} }
\newtheoremstyle{dotless}{6pt}{18pt}{}{}{\bfseries}{.}{\newline}{}
\theoremstyle{dotless}
\newtheorem{thm}{Theorem}
\newtheorem{defi}[thm]{Definition}
\newtheorem{lem}[thm]{Lemma}
\newtheorem{ex}[thm]{Example}

\newtheorem{rem}[thm]{Remark}
\newtheorem{prop}[thm]{Proposition}
\newtheorem{cor}[thm]{Corollary}
\newtheorem{fact}[thm]{Fact}
\newcommand{\changefont}[3]{\fontfamily{#1} \fontseries{#2} \fontshape{#3} \selectfont}
\changefont{ptm}{m}{n}
\usepackage{xcolor}
\definecolor{UniBlue}{RGB}{83,121,170}
\title{Discrete cyclic systems and circle congruences}
\author{Udo Hertrich-Jeromin and Gudrun Szewieczek}
\begin{document}
\maketitle
\color{black}
\noindent
\bibliographystyle{plain} 
\begin{center}
\begin{minipage}{13cm}\small
\textbf{Abstract.} We discuss integrable discretizations of
3-dimensional cyclic systems, that is, orthogonal coordinate
systems with one family of circular coordinate lines. In
particular, the underlying circle congruences are
investigated in detail, and characterized by the existence of
a certain flat connection. Within the developed framework,
discrete cyclic systems with a family of discrete flat fronts
in hyperbolic space and discrete cyclic systems, where all
coordinate surfaces are discrete Dupin cyclides,
are investigated.
\end{minipage}
\vspace*{0.5cm}\\\begin{minipage}{13cm}\small
\end{minipage}
\vspace*{0.5cm}\\\begin{minipage}{13cm}\small
\end{minipage}
\end{center}

\section{Introduction}
\noindent In \cite{rib_cyclic} Ribaucour investigated circle
congruences that admit a 1-parameter family of orthogonal
surfaces and called those congruences \emph{cyclic} (also
called \emph{normal}).
The main motivations to
study these special circle congruences seem to be twofold:
firstly, a surface family orthogonal to a cyclic circle
congruence gives rise to a special orthogonal coordinate
system (\emph{cyclic system}), where the orthogonal
trajectories of one family are circular. As a consequence,
two families of coordinate surfaces then consist of channel
surfaces. Examples are provided by orthogonal systems with
a family of surfaces that are parallel in a fixed space
form (cf \cite{cyclic_guichard_eth, salkowski}) and
special cyclidic coordinate systems (also called totally
cyclic), where all coordinate surfaces are Dupin cyclides
\cite{darboux_ortho, salkowski, totallycyclic}. Cyclic
systems widely used in physics include rotational systems
\cite{MR947546}, as for example, spherical and toroidal
coordinates.

Secondly, cyclic circle congruences can be employed to
construct (families of) surfaces of various special
types by imposing further (geometric) conditions on the
cyclic circle congruence. Amongst them are, for example,
pseudospherical surfaces related by Bianchi transformations
\cite{darboux_ortho, MR0115134} or parallel families
of flat fronts in hyperbolic space \cite{MR2652493,
MR2196639}.
The former example generalizes to a remarkable class
of cyclic circle congruences given by curved flats in
the space of circles
(so-called flat spherical or hypercyclic systems) that come
with an orthogonal family of Guichard surfaces \cite{MR239516}.
Higher dimensional analogues lead to 3-dimensional conformally
flat hypersurfaces \cite{MR1421285} and, more generally,
M\"obius flat hypersurfaces \cite{csg_45, curvedflats}.

\truepar
An integrable discretization of orthogonal coordinate
systems was given in 
\cite{org_principles, orth_nets_clifford, ddg_book, orth_cds}
where those were introduced as higher-dimensional circular
(principal) nets.
The main goal of the present work is to explore discrete
counterparts of cyclic circle congruences and their associated
cyclic coordinate systems
(see Def \ref{def_cyclic_system}
 and \ref{def_cyclic_congruence}),
based on this definition.
In this way, we anticipate
to pave the way for further studies in this context,
inspired by the rich smooth theory as sketched above.
For example, based on the observations in Subsection
 \ref{subsect_flat_fronts} we shall investigate relations
 of various approaches to flat fronts in hyperbolic
 space and, in particular, prove the existence of a
 Weierstrass-type representation for discrete flat fronts
 in \cite{disc_flat_front}.
Moreover, since the established theory naturally generalizes
 to higher-dimensional systems, it shall lead to discrete
 notions for 3-dimensional conformally flat hypersurfaces
 and M\"obius flat hypersurfaces as orthogonal surfaces of
 discrete cyclic systems stemming from discrete flat fronts.
 
\truepar
Another main goal of this paper is to further examine,
and promote, the use of discrete connections that are given
in the simplest way possible:
by the ``reflections'' of the underlying ambient geometry.
In our case,
these will be Lie and M-Lie inversions
(see Def \ref{def_M-Lie})
that can be used to generate cyclic circle congruences and
their orthogonal nets in an efficient way.
In a sometimes more implicit way, such connections have
been used in the theory of discrete orthogonal systems
for a long time,
see~\cite{org_principles, orth_nets_clifford,
 ddg_book, bcjpr, lin_weingarten, disc_cmc}.
Here, we aim to present a more explicit treatment and,
in particular, to explicitely investigate and employ
the properties of connections built from Lie and M-Lie
inversions in the context of cyclic systems.
In this way, we hope to contribute to a methodologically
systematic and transparent approach to the field.

\truepar
The paper is structured as follows. After an introductory
Section \ref{section_prelim} on basic concepts and
facts on circles and spheres that will be essential
for what follows, the main notions of the text are presented in
Section \ref{section_cyclic_systems}.
Here we approach the discretization of smooth cyclic systems
from two different angles:
firstly, we
consider discrete triply orthogonal systems that contain
two coordinate surface families of discrete channel
surfaces \cite{discrete_channel} and, therefore, have a
family of circular orthogonal trajectories. In Subsection
\ref{subsect_cyclic_congruence}, we then take a different point
of view: we start from a discrete 2-dimensional family
of circles and  investigate under what conditions there
exists a family
of orthogonal discrete Legendre maps that gives rise to a
discrete cyclic system.

We then show that the proposed discretizations reflect
properties that are well-known from the smooth theory,
such as a
constant cross-ratio of four orthogonal surfaces of a
cyclic circle congruence, and
the existence of discrete Ribaucour transformations between any
two orthogonal surfaces of a discrete cyclic system
(see Cor \ref{cor_const_cross} and \ref{cor_adjacent_rib}).

\truepar
In Section \ref{sect_from_rib} we pursue Ribaucour's original
approach and consider discrete cyclic circle congruences
associated to a Ribaucour pair of discrete Legendre maps.
In particular, we demonstrate that the circles that
intersect the spheres of a Ribaucour sphere congruence
in the point spheres of the two
envelopes orthogonally constitute a discrete cyclic circle
congruence (cf Thm \ref{thm_associated_rib_cyclic}).

As an application of the developed theory, we then
investigate discrete cyclic circle congruences constructed
from special discrete Ribaucour pairs, where at least one
of the initial nets is totally umbilic.
In this context we will discuss
parallel families of discrete flat fronts in hyperbolic space
as described in \cite{MR2946924, disc_sing}, and discrete
cyclic systems where all coordinate surfaces are discrete Dupin
cyclides (see $\S$\ref{subsect_flat_fronts} and
$\S$\ref{subsect_dupin_cyclides}).

\truepar
The concepts used and discussed throughout this work,
such as circles, orthogonal coordinate systems,
are well defined in M\"obius geometry but not in Lie sphere
geometry.
Nevertheless, we will
work in a Lie sphere geometric setup, where we fix a
M\"obius geometry as subgeometry. This will enable us to
use the elegant Lie sphere geometric descriptions for Dupin
cyclides and discrete channel surfaces
\cite{discrete_channel}
that will play a key role in our investigations.
Furthermore, enhancing methods developed
in~\cite{rib_families}, this setup enables us
to characterize discrete cyclic circle congruences by the
existence of M-Lie inversions that interchange adjacent
circles and induce a flat connection for the congruence
(cf Thm \ref{thm_flat_connection}). These maps then
provide an efficient way to construct discrete Legendre maps
orthogonal to the circle congruence and the other coordinate
surfaces of the associated discrete cyclic systems.

\truepar{\sl Acknowledgements\/}.
We would like to thank our colleagues and friends
 Joseph Cho,
 Wayne Rossman,
and
 Mason Pember
for enjoyable and helpful discussions around the subject area.
We also gratefully acknowledge financial support
of the work from
 the Austrian Science Fund FWF
through the
 Joint Project I3809 ``Geometric shape generation''.

\section{Preliminaries}\label{section_prelim}
\noindent Throughout this paper we shall work in a M\"obius
geometry, considered as a subgeometry of Lie sphere
geometry. In this section we will sketch basic concepts of
this setup and will formulate various facts on circles that
will become useful later in the text.
For more details or proofs the interested reader is referred
to the exhaustive literature in this area;
see, for example, the surveys \cite{blaschke} and
\cite{book_cecil}.

\truepar
We shall exploit the hexaspherical coordinate model of
Lie sphere geometry as introduced by Lie~\cite{L1872}
and consider the
6-dimensional vector space $\mathbb{R}^{4,2}$ endowed with a
metric of signature $(4,2)$. The \emph{projective light cone}
will be denoted by
\begin{equation*}
 \LL :
  = \{ \lspann{\mathfrak{s}} \subset \mathbb{R}^{4,2} \ | \
  \lspan{\mathfrak{s}, \mathfrak{s}}=0 \}
  \subset \mathbb{P}(\mathbb{R}^{4,2}) 
\end{equation*}
and represents the set of oriented 2-spheres.
Two spheres $s_1$ and $s_2$
are in oriented contact if and only if any two corresponding
vectors $\mathfrak{s}_1$ and $\mathfrak{s}_2$ in the light cone
are orthogonal.
Hence the set of contact elements,
that is, of pencils of $2$-spheres in oriented contact,
is represented by the set of lines in the Lie quadric $\LL$ or,
equivalently,
the subset of null $2$-planes in $\mathbb{R}^{4,2}$ of the
Grassmannian.

\truepar
\textbf{Convention.}\,Throughout this work, homogeneous
coordinates of elements in the projective space
$\mathbb{P}(\mathbb{R}^{4,2})$ will be denoted by the
corresponding black letter; if statements hold for arbitrary
homogeneous coordinates we will use this convention without 
explicitly mentioning it.

\truepar
To obtain a M\"obius geometry of oriented spheres as
a subgeometry in this setup, we fix a point sphere complex
$\mathfrak{p} \in \mathbb{R}^{4,2}$, $\lspan{\mathfrak{p},
\mathfrak{p}}=-1$, and recover points, that is, spheres with
radius zero, as elements in
\begin{equation*}
 \mathbb{P}(\mathcal{P}) :
  = \mathbb{P}({\mathcal L}\cap\{\mathfrak{p}\}^\perp).
\end{equation*}
The group of M\"obius transformations is then provided by
all Lie sphere transformations that preserve the point
sphere complex $\mathfrak{p}$. In particular, those preserve
the (unoriented) angle $\varphi$ between two spheres $u,v
\in \LL$, given by
\begin{equation}\label{equ_intersection_angle}
 \cos \varphi
  = 1 - \frac{
   \lspan{\mathfrak{u},\mathfrak{v}}
   \lspan{\mathfrak{p},\mathfrak{p}}
   }{
   \lspan{\mathfrak{u},\mathfrak{p}}
   \lspan{\mathfrak{v},\mathfrak{p}}
   }.
\end{equation}
Furthermore, by choosing a vector $\mathfrak{q}\in
\mathbb{R}^{4,2}\setminus \{ 0 \}$, $\lspan{\mathfrak{p},
\mathfrak{q}}=0$, we distinguish
a \emph{quadric of constant curvature}
\begin{equation*}
 \QQ :
  = \{ \mathfrak{n} \in \mathbb{R}^{4,2} \ | \
  \lspan{\mathfrak{n}, \mathfrak{n}}=0, \lspan{\mathfrak{n},
   \mathfrak{q}}=-1, \lspan{\mathfrak{n}, \mathfrak{p}}=0 \},
\end{equation*}
with constant sectional curvature
$-\lspan{\mathfrak{q}, \mathfrak{q}}$,
and obtain its complex of \emph{hyperplanes}
\begin{equation*}
 \PP :
  = \{ \mathfrak{n} \in \mathbb{R}^{4,2} \ | \
   \lspan{\mathfrak{n}, \mathfrak{n}}=0, \lspan{\mathfrak{n},
   \mathfrak{q}}=0, \lspan{\mathfrak{n}, \mathfrak{p}}=-1 \}.
\end{equation*}

\bigskip 

\noindent In Lie sphere geometry,
any element $a \in \mathbb{P}(\mathbb{R}^{4,2})$ defines
a \emph{linear sphere complex}
$\mathbb{P}(\mathcal{L}\cap\{a\}^\perp)$,
that is, a 3-dimensional family of 2-spheres.
We distinguish three types of linear sphere complexes:
if $\lspan{ \mathfrak{a}, \mathfrak{a}} = 0$, the complex
is called \emph{parabolic} and consists of all spheres
that are in oriented contact with the sphere represented
by~$a$. If $\lspan{ \mathfrak{a}, \mathfrak{a}} < 0$, we
say that the complex is \emph{hyperbolic} and for $\lspan{
\mathfrak{a}, \mathfrak{a}} > 0$ we obtain an \emph{elliptic}
linear sphere complex.

In M\"obius geometry, that is, for a fixed point sphere
complex $\mathfrak{p}$, the latter has a beautiful geometric
characterization: the linear sphere complex then contains
all spheres that intersect the two spheres (that coincide
up to orientation)
\begin{equation}\label{equ_sphere_compl}
 \mathfrak{s}_a^\pm \in \lspann{\mathfrak{a}, \mathfrak{p}}
\end{equation}
at the constant angle
\begin{equation*}
 \cos^2 \varphi = \frac{K}{K-1}, \ \ \text{where } \ 
 K = \frac{\lspan{\mathfrak{a}, \mathfrak{p}}^2}
  {\lspan{\mathfrak{a}, \mathfrak{a}}
   \lspan{\mathfrak{p}, \mathfrak{p}}}.
\end{equation*}
In particular, the spheres in a linear sphere complex
intersect the spheres $s_a^\pm$ orthogonally if and only if
$\lspan{\mathfrak{a}, \mathfrak{p}}=0$.

Conversely, suppose that $s \in \LL$ is a sphere,
$\mathfrak{s}\not\perp\mathfrak{p}$;
then the elliptic linear sphere complex that contains
all spheres intersecting $s$ orthogonally is given by
\begin{equation}\label{equ_orth_complex}
 \mathfrak{a} :
  = \mathfrak{s}+\lspan{\mathfrak{s},\mathfrak{p}}\mathfrak{p}.
\end{equation}
Hence, a sphere $t \in \LL$ intersects the sphere $s$
orthogonally if and only if $\lspan{\mathfrak{t},
\mathfrak{a}}=0$.

\truepar
Any elliptic and hyperbolic linear sphere complex may be
used to define a reflection:
let $a \in \mathbb{P}(\mathbb{R}^{4,2})$, $\lspan{
\mathfrak{a}, \mathfrak{a}} \neq 0$, then the \emph{Lie
inversion with respect to the linear sphere complex
$\mathbb{P}(\mathcal{L}\cap\{a\}^\perp)$} is given by
\begin{equation*}
 \sigma_a:\mathbb{R}^{4,2} \rightarrow \mathbb{R}^{4,2}, \ \ 
 \mathfrak{r} \mapsto \sigma_a(\mathfrak{r}) :
  = \mathfrak{r}-\frac{2\lspan{\mathfrak{r}, \mathfrak{a}}}
   {\lspan{\mathfrak{a}, \mathfrak{a}}}\mathfrak{a}.
\end{equation*}
Any Lie inversion is an involution that maps spheres
to spheres and preserves oriented contact between
spheres. Moreover, we emphasize that a Lie inversion $\sigma_a$
preserves all elements that lie in the corresponding linear
sphere complex
$\mathbb{P}(\mathcal{L}\cap\{a\}^\perp)$. 
In particular, Lie inversions that preserve the point
sphere complex will play a crucial role:
\begin{defi}\label{def_M-Lie}
A Lie inversion $\sigma_a$ that preserves the point
sphere complex, $\lspan{\mathfrak{p}, \mathfrak{a}}=0$,
will be called an \emph{M-Lie inversion}.
\end{defi}

\truepar
Clearly, any M-Lie inversion is a M\"obius
transformation and generalizes the concept of M\"obius
inversions: if $a$ determines an elliptic linear sphere
complex, the M-Lie inversion becomes a M\"obius inversion,
that is, it provides a reflection in the spheres $s^\pm_a$
as given in (\ref{equ_sphere_compl}). However, if the
corresponding linear sphere complex is hyperbolic,
it can be thought of as an antipodal map.

Note that the M-Lie inversion $\sigma_\mathfrak{p}$ with
respect to the point sphere complex $\mathfrak{p}$ reverses
the orientation of all spheres.

\truepar
Let $s_1, s_2 \in \LL$ be two spheres that are not in
oriented contact and $\mathfrak{s}_1, \mathfrak{s}_2 \in \LLL$
fixed homogeneous coordinates, then the Lie inversion with
respect to $\mathfrak{a}:=\mathfrak{s}_1-\mathfrak{s}_2$,
interchanges the spheres $s_1$ and $s_2$. However, note that
different choices for the homogeneous coordinates provide
a 1-parameter family of Lie inversions.
When $s_1$ and $s_2$ are not point spheres,
$\mathfrak{s}_1,\mathfrak{s}_2\not\perp\mathfrak{p}$,
then there is a unique M-Lie inversion, determined by
\begin{equation*}
 \mathfrak{a} :
 = \lspan{\mathfrak{s}_2, \mathfrak{p}}\mathfrak{s}_1
 - \lspan{\mathfrak{s}_1, \mathfrak{p}}\mathfrak{s}_2,
\end{equation*}
that maps $s_1$ to $s_2$ and preserves the point sphere
complex~$\mathfrak{p}$.

\truepar
For later reference, we also recall \cite{blaschke,
rib_families} that the formula for the cross-ratio of
four spheres which are pairwise related by a Lie inversion
simplifies: let $s_1, s_2 \in \LL$, $\lspan{\s_1, \s_2}\neq
0$, be two spheres that are not contained in the linear
sphere complex
$\mathbb{P}(\mathcal{L}\cap\{a\}^\perp)$,
then
\begin{align}\label{equ_formula_cross_ratio}
 \text{cr}(s_2, \sigma_a(s_2), \sigma_a(s_1), s_1)
  = 2\frac{\lspan{\s_1,\fa}\lspan{\s_2,\fa}}
   {\phantom{2}\lspan{\fa,\fa}\lspan{\s_1,\s_2}}.
\end{align}
Using appropriate M-Lie inversions, this formula can also
be used to compute the cross-ratio between four concircular
point spheres.

\subsection{Circles in this framework}
Throughout this text we will consider unoriented circles,
thus objects that belong to M\"obius geometry. Hence, in
the employed Lie geometric framework,
we again fix a point sphere complex $\mathfrak{p}$ to
distinguish a M\"obius subgeometry.
Then circles arise as special Dupin cyclides,
where one family of curvature spheres are point spheres.

A \emph{circle} $\Gamma$ is provided by an orthogonal
splitting of $\mathbb{R}^{4,2}$,
\begin{equation*}
 \Gamma = (\gamma, \gamma^\perp) \in G_{(2,1)}^{\mathcal{P}}
 \times G_{(2,1)}, 
\end{equation*}
where $G_{(2,1)}^{\mathcal{P}}$ denotes the set of all
$(2,1)$-planes that are orthogonal to the point sphere
complex~$\mathfrak{p}$. Therefore, all spheres contained
in $\gamma \in G_{(2,1)}^{\mathcal{P}}$ are point spheres,
namely, the points of the circle.
The spheres in $\gamma^\perp$
are in oriented contact with all circle points, hence,
provide a M\"obius geometric pencil of spheres.

If we additionally fix a vector $\mathfrak{q} \in
\mathbb{R}^{4,2}\setminus \{ 0 \}$, $\lspan{\mathfrak{q},
\mathfrak{p}}=0$, to distinguish a space form $\QQ$, we
obtain \emph{lines} in this space form as special circles
satisfying $\mathfrak{q} \in \gamma$.

\begin{fact}\label{fact_intersection_spheres}
If two spheres $s_1, s_2 \in \LL$ intersect
(in the M\"obius geometry given by $\mathfrak{p}$),
then their circle of intersection is given by
\begin{equation*}
 \Gamma = (\gamma, \gamma^\perp)
  \in G_{(2,1)}^{\mathcal{P}} \times G_{(2,1)},
   \ \text{where }
  \gamma^\perp:=\lspann{\mathfrak{s}_1, \mathfrak{s}_2,
   \mathfrak{p}}. 
\end{equation*}
\end{fact}

\noindent A crucial concept in the study of cyclic systems
will be spheres and circles that intersect orthogonally. In
what follows, we summarize several useful constructions in
this realm and formulate them in the Lie sphere geometric
framework.

Firstly, note that a circle $\Gamma=(\gamma, \gamma^\perp)$
intersects a sphere $s \in \LL$ orthogonally if and only if
$s$ intersects all spheres in $\gamma^\perp$ orthogonally. In
fact, a weaker condition is sufficient to ensure orthogonality:
\begin{fact} \label{fact_two_orth}
Let $s_1, s_2 \in \LL$ be two spheres that intersect in the
circle $\Gamma = (\gamma, \gamma^\perp)$. Then, $\Gamma$
intersects a sphere $t \in \LL$ orthogonally if and only if
the sphere $t$ intersects $s_1$ and $s_2$ orthogonally.
\end{fact}
\begin{proof}
Suppose that the sphere $t \in \LL$ intersects $s_1$ and $s_2$
orthogonally, that is,
\begin{equation*}
 \lspan{\mathfrak{t}, \mathfrak{s}_i + \lspan{\mathfrak{s}_i,
  \mathfrak{p}}\mathfrak{p}}=0 \ \ \text{for } \ i=\{ 1,2 \}.
\end{equation*} 
Then, a straightforward computation shows that $t$ intersects
any sphere in
\begin{equation*}
 \gamma^\perp
  = \lspann{\mathfrak{s}_1, \mathfrak{s}_2, \mathfrak{p}}
\end{equation*}
orthogonally, which proves the claim.
\end{proof}

\noindent Spheres that are orthogonal to a fixed circle
satisfy the following properties, which are also illustrated
in Figure \ref{fig:scheme_cyclic}: 
\begin{fact}\label{fact_orth_cont_els}
The spheres that orthogonally intersect a circle $\Gamma$
in a fixed point $m\subset\gamma$
of the circle lie in two
contact elements $m \in f_m, \tilde{f}_m \in \mathcal{Z}$
that coincide up to orientation, that is, $\tilde{f}_m =
\sigma_{\mathfrak{p}}(f_m)$.

Moreover, for any two points $m,n\subset\gamma$
of the circle, the associated contact elements $f_m,
\tilde{f}_m, f_n$ and $\tilde{f}_n$ of orthogonal spheres
pairwise share a common sphere, that is,
\begin{equation*}
 f_m \cap f_n \neq \{ 0 \} \ \ \ \text{ or } \ \ 
 f_m \cap \sigma_{\mathfrak{p}}(f_n) \neq \{ 0 \}.
\end{equation*}
\end{fact}
\begin{proof}
Let $\Gamma= (\gamma, \gamma^\perp)$ be a circle,
given as the orthogonal intersection of two spheres
$s_1,s_2\in\LL$,
and let $\mathfrak{m}\in\gamma$ represent a point
of this circle.
Without loss of generality, we choose
homogeneous coordinates $\mathfrak{s}_i \in s_i$ such that
$\lspan{\mathfrak{s}_i, \mathfrak{p}}=-\lspan{\mathfrak{s}_1,
\mathfrak{s}_2}=1$ for $i =\{ 1, 2\}$.

Then, as a consequence of Fact \ref{fact_two_orth}, all
spheres that intersect the circle $\Gamma$ in the point $m$
orthogonally lie in the subspace
\begin{equation}\label{equ_two_contact_orth}
 \mathcal{O}_m := \lspann{\mathfrak{s}_1 + \mathfrak{p}, \  
 \mathfrak{s}_2 + \mathfrak{p}, \ \mathfrak{m}}^\perp.
\end{equation}
Since the subspace $\mathcal{O}_m^\perp$ has signature
$(++0)$, we conclude that for any point of the circle, the
sought-after orthogonally intersecting spheres lie in two
contact elements that coincide up to orientation.

Furthermore, since spheres that contain the points represented
by $\mathfrak{m}, \mathfrak{n} \in \gamma$ and that
are orthogonal to $\Gamma$ lie in the subspace
$\lspann{\mathfrak{s}_1 + \mathfrak{p}, \
\mathfrak{s}_2 + \mathfrak{p}, \ \mathfrak{m}, \ 
\mathfrak{n}}^\perp$, the second claim follows.
\end{proof}

\begin{fact}\label{fact_circles_on_sphere}
Let $\Gamma_1$ and $\Gamma_2$ be two circles and denote by
$\mathfrak{s}_i, \mathfrak{t}_i \in \gamma_i^\perp \cap \LLL$,
$i=\{1,2\}$, two spheres that determine the corresponding
circles $\Gamma_i$. Then, the two circles lie on a common
sphere if and only if the subspace
\begin{equation*}
 \mathcal{S} := \lspann{\mathfrak{s}_1, \mathfrak{s}_2,
 \mathfrak{t}_1, \mathfrak{t}_2, \mathfrak{p}}
 \subset \mathbb{R}^{4,2}
\end{equation*}
is at most 4-dimensional.
\end{fact}
\begin{proof}
Suppose that the two circles lie on the sphere $k \in
\LL$. Then, by Fact \ref{fact_intersection_spheres}, we
conclude that $\mathcal{S}$ is at most 4-dimensional.

To show the converse, we choose, without loss of generality,
two spheres $\mathfrak{s}_1, \mathfrak{t}_1 \in \gamma_1^\perp
\cap \LLL$ that intersect orthogonally and homogeneous
coordinates such that $-\lspan{\mathfrak{s}_1,
\mathfrak{t}_1}=\lspan{\mathfrak{s}_1,
\mathfrak{p}}=\lspan{\mathfrak{t}_1, \mathfrak{p}}=1$.
If $\mathcal{S}$ is at most 4-dimensional, then there exist
constants $\lambda_i \in \mathbb{R}$ such that
\begin{equation*}\label{equ_dependent}
 \lambda_1 \mathfrak{s}_1 + \lambda_2 \mathfrak{t}_1 +
 \lambda_3 \mathfrak{s}_2 + \lambda_4 \mathfrak{t}_2 +
 \mathfrak{p}=0.
\end{equation*}
Thus, by setting $c^\pm:= (\lambda_1+\lambda_2) \pm
\sqrt{\lambda_1^2+\lambda_2^2}$ and $d^\pm:=1-c^\pm$, we
obtain two spheres (with opposite orientation)
\begin{equation*}
 \lambda_1 \mathfrak{s}_1 + \lambda_2 \mathfrak{t}_1
 + c^\pm \mathfrak{p}
 = -\lambda_3 \mathfrak{s}_2 - \lambda_4 \mathfrak{t}_2
 -  d^\pm\mathfrak{p}
\end{equation*}
that lie in $\gamma_1^\perp \cap \gamma_2^\perp$ and,
therefore, contain the two circles $\Gamma_1$ and
$\Gamma_2$. This proves the claim.
\end{proof}  

\begin{fact}\label{fact_orth_circle}
Given a sphere $r \in \LL$ and two point spheres $p_1,
p_2 \in \PL$ lying on it, the circle $\Gamma$ that
intersects the sphere $r$ orthogonally and passes through
the points $p_1$ and $p_2$ is described by the $(2,1)$-plane
\begin{equation}\label{equ_orth_circle}
 \gamma :
  = \lspann{\mathfrak{p}_1, \ \mathfrak{p}_2,
  \ \mathfrak{r}+\lspan{\mathfrak{r},
  \mathfrak{p}}\mathfrak{p}} \in G^{\mathcal{P}}_{(2,1)}.
\end{equation}
\end{fact}
\begin{proof}
Firstly, note that $\gamma$ is a $(2,1)$-plane orthogonal
to the point sphere complex $\mathfrak{p}$. Thus, $\gamma$
indeed describes a circle by $\Gamma := (\gamma, \gamma^\perp)
\in G_{(2,1)}^{\mathcal{P}} \times G_{(2,1)}$. Furthermore,
suppose that $\mathfrak{s} \in \gamma^\perp \cap \LLL$
and consider the corresponding linear sphere complex with
orthogonal intersection angle, that is,
\begin{equation*}
 \mathfrak{a} :
  = \mathfrak{s} + \lspan{\mathfrak{s},
  \mathfrak{p}}\mathfrak{p}.
\end{equation*} 
Then, because of $\lspan{\mathfrak{s}, \mathfrak{r}
 + \lspan{\mathfrak{r}, \mathfrak{p}}\mathfrak{p}}=0$,
we conclude that
\begin{equation*}
 \lspan{\mathfrak{r}, \mathfrak{a}}=0.
\end{equation*}
Since this holds for any sphere in $\gamma^\perp \cap \LLL$,
the constructed circle $\Gamma$ intersects the sphere $r$
orthogonally in the points $p_1$ and $p_2$.
\end{proof}

\noindent As a consequence we obtain:
\begin{fact} 
Given a sphere $s \in \LL$ and two circles $\Gamma_1$ and
$\Gamma_2$ that intersect the sphere orthogonally in the
points $p_1^1, p_2^1$ and $p_1^2, p_2^2$.
Then the circles lie on a common sphere if and only if
the four points are concircular, that is,
any homogeneous coordinate vectors $\mathfrak{p}_i^j$
are linearly dependent.
\end{fact}
\begin{proof} 
The sought-after sphere has to lie in $\gamma_1^\perp
\oplus \gamma_2^\perp$, where $\gamma_i$ is given by
(\ref{equ_orth_circle}). Thus, the subspace $\gamma_1 \oplus
\gamma_2$ has to be a 4-dimensional space.
\end{proof}

\subsection{Ribaucour transformations between two circles}
Recall \cite{rib_coords, discrete_channel} that any two
cospherical circles are related by a Ribaucour transformation,
that is, they envelop a common circle congruence. For two
unparametrized circles on a sphere, there are at most two
Ribaucour transformations that induce different Ribaucour
correspondences between the point spheres of the two circles.

Those point-to-point mappings can be described by two M-Lie
inversions:
\begin{lem}\label{lem_rib_trafo_lie_inversion}
Let $\Gamma_1$ and $\Gamma_2$ be two circles that are not
tangent to each other and lie on the sphere $k \in \LL$. Then
the two Ribaucour correspondences between $\Gamma_1$ and
$\Gamma_2$ are induced by the two M-Lie inversions $\sigma_a$
and $\sigma_{\tilde{a}}$ determined by the linear sphere
complexes
\begin{equation*}
 \mathfrak{a} :
  = \lspan{\mathfrak{s}_2, \mathfrak{p}} \mathfrak{s}_1
  - \lspan{\mathfrak{s}_1, \mathfrak{p}} \mathfrak{s}_2
   \ \ \ \text{and } \ \ 
 \tilde{\mathfrak{a}} :
  = \lspan{\sigma_\mathfrak{p}(\mathfrak{s}_2), \mathfrak{p}}
   \mathfrak{s}_1 - \lspan{\mathfrak{s}_1,
   \mathfrak{p}}\sigma_\mathfrak{p}(\mathfrak{s}_2),
\end{equation*}
where $\mathfrak{s}_i \in \gamma_i^\perp$, $i= \{ 1,2 \}$,
are two spheres that intersect $k$ orthogonally.
\end{lem}
\noindent We remark that, for two touching circles
$\Gamma_1$ and $\Gamma_2$ on a common sphere $k \in
\LL$, there exists only one Ribaucour correspondence:
in this case, either $\lspan{\mathfrak{s}_1,
\mathfrak{s}_2}=0$ or $\lspan{\mathfrak{s}_1,
\sigma_\mathfrak{p}(\mathfrak{s}_2)}=0$ and, therefore,
one of the M-Lie inversions described in Lemma
\ref{lem_rib_trafo_lie_inversion} degenerates. Thus,
between two touching circles we obtain a unique Ribaucour
correspondence.
\begin{proof}
Without loss of generality, we assume that 
\begin{equation*}
 \lspan{\mathfrak{s}_1, \mathfrak{p}}=\lspan{\mathfrak{s}_2,
 \mathfrak{p}}=\lspan{\mathfrak{k}, \mathfrak{p}}=1.
\end{equation*}
Then, since $s_1$ and $s_2$ intersect the sphere $k$
orthogonally, from (\ref{equ_intersection_angle}) we
conclude that
\begin{equation*}
 \lspan{\mathfrak{s}_1, \mathfrak{k}}=\lspan{\mathfrak{s}_2,
 \mathfrak{k}}=-1.
\end{equation*}
Hence, the Lie inversions $\sigma_a$ and $\sigma_{\tilde{a}}$
preserve the sphere $k$, as well as the point sphere
complex~$\mathfrak{p}$. Furthermore, they map $s_1$ to
$s_2$ and $\lspann{\sigma_\mathfrak{p}(\mathfrak{s}_2)}$,
respectively. Therefore, they are M-Lie inversions and
interchange point spheres of the circles $\Gamma_1$ and
$\Gamma_2$ with each other.

For any induced pair of point spheres $p_1$ and $p_2$,
there exists a sphere $t \in \LL$ that contains these points
and is in oriented contact with $s_1$ and $s_2$. Then, $t$
is preserved by those M-Lie inversions and intersects $k$
orthogonally.

Moreover, the $(2,1)$-plane $\lspann{\mathfrak{t},
\mathfrak{k}, \mathfrak{p}}$ determines a circle that is
tangent to the circles $\Gamma_1$ and $\Gamma_2$ at the
points $p_1$ and $p_2$. This proves the claim.
\end{proof}

\noindent The generic ambiguity of the Ribaucour
correspondence between two cospherical circles is eliminated
by the choice of one admissible point sphere pair:
\begin{cor}\label{cor_choice_ribtrafo}
Let $\Gamma_1$ and $\Gamma_2$ be two cospherical circles
and let $\mathfrak{p}_i \in \gamma_i \cap \LLL$, $i=\{1,2\}$,
represent a pair of points that is contained in a circle tangent
to $\Gamma_1$ and $\Gamma_2$. Then, there exists a unique
Ribaucour transformation between the circles that extends
the correspondence between the points $p_1$ and  $p_2$.
\end{cor}
\begin{proof}
Let $\Gamma_1$ and $\Gamma_2$ be two circles lying on the
sphere $k \in \LL$. Moreover, we denote by $\mathfrak{s}_1
\in \gamma_1^\perp \cap \LLL$ a sphere that is orthogonal
to $k$ and define the sphere
\begin{equation*}
 \mathfrak{s}_{12} :
  = \lspan{\mathfrak{p}_1, \mathfrak{p}_2} \mathfrak{s}_1
  - \lspan{\mathfrak{s}_1, \mathfrak{p}_2}\mathfrak{p}_1.
\end{equation*}
The sphere $s_{12}$ contains the point pair $(p_1, p_2)$
and lies in the contact element $\lspann{s_1, p_1}$. Then
the sphere $\mathfrak{s}_2 \in \gamma_2^\perp \cap
\lspann{\mathfrak{s}_{12}, \mathfrak{p}_2}$ provides the
construction for the sought-after Ribaucour correspondence:
the M-Lie inversion determined by the linear sphere complex
\begin{equation*}
 \mathfrak{a} :
 = \lspan{\mathfrak{s}_2, \mathfrak{p}}\mathfrak{s}_1
 - \lspan{\mathfrak{s}_1, \mathfrak{p}} \mathfrak{s}_2 
\end{equation*}
induces the Ribaucour correspondence that maps $p_1$
onto $p_2$.
\end{proof}

\noindent Conversely, if two circles $\Gamma_1$ and $\Gamma_2$
are related by an M-Lie inversion~$\sigma$, then $\sigma$
induces a smooth Ribaucour transformation between the two
circles:  let
\begin{equation*}
 t \mapsto (\mathfrak{c}_1(t), \ \sigma_a(\mathfrak{c}_1(t))
\end{equation*}
be a simultaneous parametrization of the two circles, where
$\mathfrak{c}_1(t) \in \gamma_1 \cap \LLL$. For any  sphere
$\mathfrak{s}_1 \in \gamma_1^\perp$, the map
\begin{equation*}
 t \mapsto \mathfrak{s}(t) :
  = \lspann{\mathfrak{s}_1, \mathfrak{c}_1(t)}
  \cap \lspann{\sigma_a(\mathfrak{s}_1),
  \sigma_a(\mathfrak{c}_1(t)} \cap \mathcal{L}
\end{equation*}
defines a 1-parameter family of spheres that is in
oriented contact with both circles. Hence, by \cite[Thm
2.15]{discrete_channel}, we conclude that $t \mapsto s(t)$
gives rise to a Dupin cyclide with curvature circles
$\Gamma_1$ and $\Gamma_2$, which proves the claim.

\section{Discrete cyclic systems}\label{section_cyclic_systems}
\noindent In this section we develop the main notions of
this paper from two different points of view: firstly, we
consider discrete cyclic systems, that is, discrete orthogonal
coordinate system that have two families of discrete channel
surfaces as coordinate surfaces. Secondly, we investigate
under which conditions a discrete circle congruence is
cyclic, thus, is the underlying circle congruence of a
discrete cyclic system.

\truepar
\textbf{Notation.}
Throughout the text we will adopt the following notation
conventions for domains of discrete maps:
we will consider a simply connected subset of the lattice
$\mathbb{Z}^3$, organized into
 vertices $\bar{\mathcal{V}}$,
 edges $\bar{\mathcal{E}}$, and
 faces $\bar{\mathcal{F}}$;
for a $2$-dimensional ``slice'', modeled in $\mathbb{Z}^2$,
we will use
 $\mathcal{V}$,
 $\mathcal{E}$, and
 $\mathcal{F}$,
for the sets of vertices, edges, and faces, respectively.
Thus our domains will be (rather trivial) quadrilateral
or cubical cell complexes,
where only cells of the dimensions $0$, $1$, and $2$
will play a role.
Furthermore, $I\subset\mathbb{Z}$ will denote
(the vertex set of) a discrete (closed) interval.

The domains under consideration are assumed to be sufficiently
large, that is, in each coordinate direction there exist at
least three vertices.

\subsection{Definition and basic properties}
In the spirit of \cite{org_principles, ddg_book}, we consider
discrete triply orthogonal systems as principal contact
element nets:
\begin{defi}\label{def_disc_tops}
A \emph{discrete triply orthogonal system} is a map 
\begin{equation*}
 f:\bar{\mathcal{V}} \rightarrow \mathcal{Z}\times \mathcal{Z}
 \times \mathcal{Z}, \ i \mapsto f_i= (f_i^1,f_i^2,f_i^3)
\end{equation*}
such that
\begin{itemize}
\item[(i)] at each vertex the point spheres
 of the contact elements coincide,
\begin{equation*}
 f_i^1 \cap f_i^2 \cap f_i^3 =: f_i^p \in \PL,
\end{equation*} 
\item[(ii)] at each vertex the spheres of different
 contact elements intersect orthogonally and
\item[(iii)] any two adjacent contact elements of
 the same family intersect,
\begin{equation*}
 f_i^\mu \cap f_j^\mu =: s_{ij}^\mu \in \LL, \ 
 \mu =\{ 1,2,3 \}.
\end{equation*} 
\end{itemize}
The map $f^p: \bar{\VV} \rightarrow \PL$ will be called the
\emph{point sphere map} of $f$ and provides a 3-dimensional
circular net.
\end{defi} 

\noindent We remark that a map $f:\bar{\mathcal{V}} \rightarrow
\mathcal{Z}\times \mathcal{Z} \times \mathcal{Z}$ satisfies
conditions (i) and (ii) if and only if,
at any vertex $i\in\bar{\mathcal{V}}$,
any two spheres $s_i^\lambda$ and $s_i^\mu$ of any two
different contact elements $f_i^\lambda$ and $f_i^\mu$ fulfill
\begin{equation*}
 \lspan{\mathfrak{s}_i^\lambda, \mathfrak{s}_i^\mu
 + \lspan{\mathfrak{p}, \mathfrak{s}_i^\mu}\mathfrak{p}}=0.
\end{equation*}
Note that this condition is symmetric
in $\lambda,\mu\in\{1,2,3\}$, $\lambda\neq\mu$.

\truepar 
Any discrete triply orthogonal system consists of three
(Lam\'e) families of \emph{coordinate surfaces},
\begin{align*}
 \lambda &\mapsto \{\bar{\mathcal{V}} \ni (\lambda, x_2, x_3)
  \mapsto f_{(\lambda, x_2, x_3)}^1 \in \mathcal{Z} \}, \\
 \lambda &\mapsto \{\bar{\mathcal{V}} \ni (x_1, \lambda, x_3)
  \mapsto f_{(x_1, \lambda,x_3)}^2 \in \mathcal{Z} \}, \\
 \lambda &\mapsto \{\bar{\mathcal{V}} \ni (x_1, x_2, \lambda)
  \mapsto f_{(x_1, x_2,\lambda)}^3 \in \mathcal{Z} \};
\end{align*}
those are discrete Legendre maps and represent discrete
surfaces that intersect orthogonally along discrete
curvature lines. We call those discrete lines of intersection
\emph{$x_i$-trajectories} of the system.


\truepar

Due to condition (iii), any two adjacent coordinate surfaces
of the same family are related by a discrete  Ribaucour
transformation (cf \cite{org_principles, ddg_book,
rib_families}).

In particular, for any edge of a triply orthogonal system,
we obtain three distinguished spheres $(s_{ij}^\mu)_\mu$
that pairwise intersect orthogonally:
two of these spheres are curvature spheres of coordinate
surfaces and the third sphere is a Ribaucour sphere enveloped
by the Ribaucour pair of adjacent coordinate surfaces.

\truepar
Imitating the smooth notion of a cyclic system
(cf \cite{coolidge, darboux_ortho, rib_cyclic}),
we introduce the following notion:
\begin{defi}\label{def_cyclic_system}
A \emph{discrete cyclic system} is a discrete triply
orthogonal system such that two families of coordinate
surfaces are discrete channel surfaces that intersect along
their circular curvature lines.
\end{defi}

\noindent We emphasize that the notion of a discrete
triply orthogonal system as well as cyclicity of it are
only invariant under M\"obius transformations, hence, those
notions depend on the choice of a point sphere complex.

\truepar
Recall that a discrete channel surface in the sense of
\cite{discrete_channel} is a discrete Legendre map that
admits a constant Lie cyclide for each coordinate ribbon
of one coordinate direction. Hence, the curvature spheres
of the discrete channel surface along each coordinate ribbon
are curvature spheres of this constant Lie cyclide.

As a consequence, a discrete channel surface has
one family of circular curvature lines and one family of
curvature spheres is constant along each of them. Moreover,
any two adjacent circular curvature lines are curvature
lines of the corresponding (constant) Lie cyclide.

Thus, since one family of (orthogonal) trajectories of
a discrete cyclic system provides the circular curvature
lines of the discrete channel surfaces, we conclude:

\begin{cor}
One family of trajectories of a discrete cyclic system are
circular.
\end{cor}

\noindent Moreover, as circular curvature lines of a
discrete channel surface, two adjacent trajectories are
related by a discrete Ribaucour transformation that is induced
by a smooth Ribaucour transformation (cf \cite[Prop
2.10]{discrete_channel}). Note that this correspondence
is given by an M-Lie inversion as described in Lemma
\ref{lem_rib_trafo_lie_inversion}.

Hence:
\begin{cor}
A discrete cyclic system provides a discrete 2-dimensional
congruence of circles where any two adjacent circles are
cospherical.
\end{cor}

\noindent Furthermore, for any discrete cyclic system we
obtain two distinguished 2-parameter families of Dupin cyclides,
namely, those that provide constant Lie cyclides for the
discrete channel surfaces.

Any circle of the above congruence is then a curvature line on
four Dupin cyclide patches, the Lie cyclides of the two discrete
channel surfaces that intersect along this circle. Since their
contact elements along the circle intersect orthogonally,
the two constant curvature spheres are mutually quer-spheres
of the smooth Lie cyclide patches.

\subsection{Cyclic circle congruences}\label{subsect_cyclic_congruence}
In this subsection, we investigate whether a given discrete
2-dimensional circle congruence admits orthogonal surfaces
and, subsequently, gives rise to a discrete cyclic system.

Thus, we consider a circle congruence 
on a 2-dimensional domain,
\begin{equation*}
 \Gamma = (\gamma, \gamma^\perp):
  \VV \rightarrow G_{(2,1)}^{\mathcal{P}} \times G_{(2,1)}.
\end{equation*} 
\begin{defi}
A discrete Legendre map $f: \VV \rightarrow \mathcal{Z}$
is said to be \emph{orthogonal} to $\Gamma$ if, for each
vertex $i \in \VV$,
\begin{itemize}
\item[(i)] the point sphere $f_i^\mathfrak{p}$ lies
 on the circle $\Gamma_i$ and 
\item[(ii)] any other sphere in the contact element
 $f_i$ intersects the circle $\Gamma_i$ orthogonally.
\end{itemize}
\end{defi}
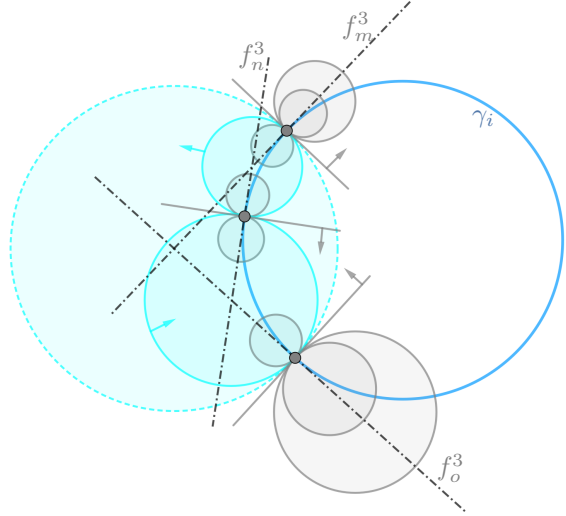
\begin{figure}[t]
\hspace*{-3cm}\begin{minipage}{6cm}
  \begin{overpic}[scale=.2]{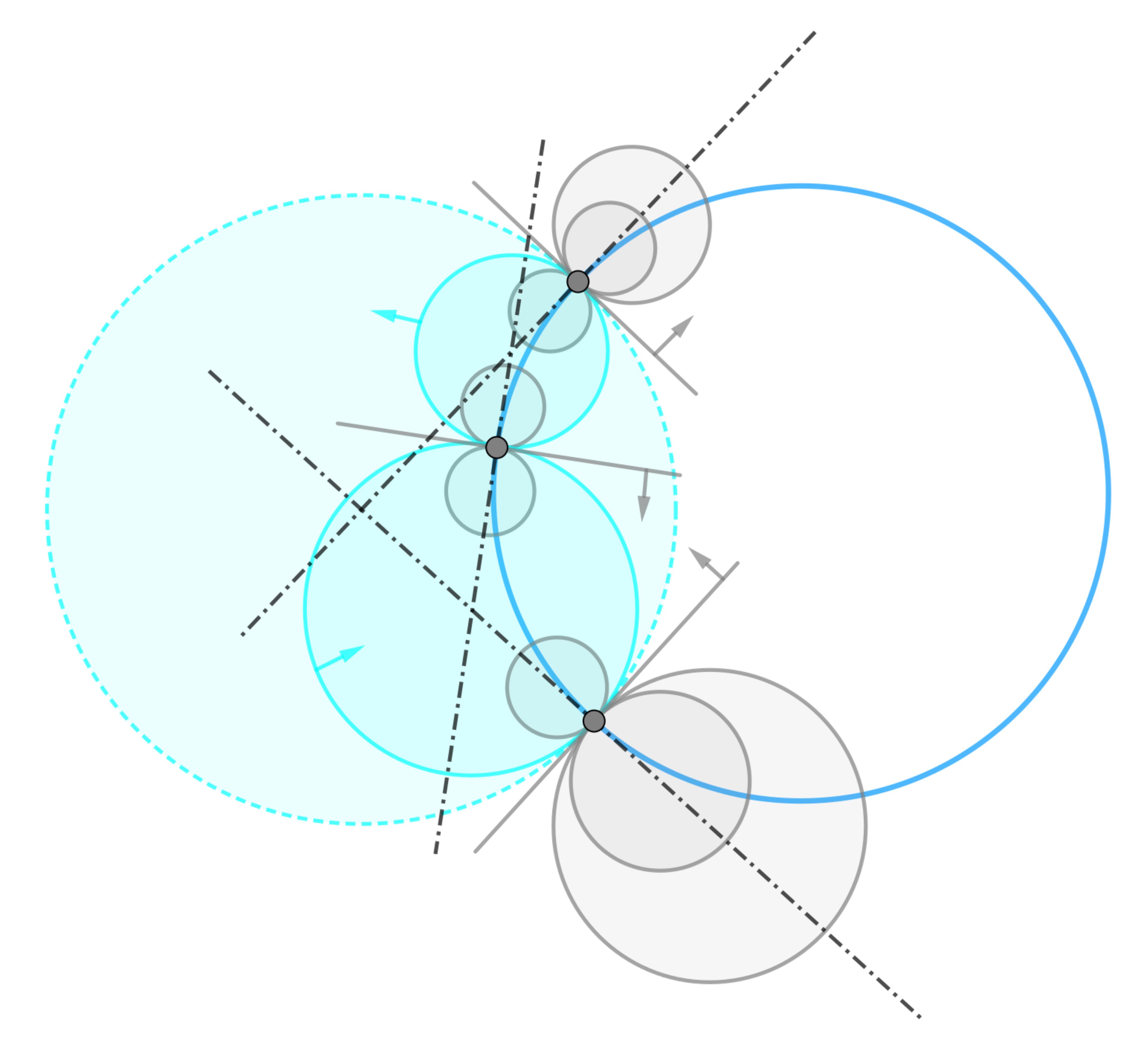}
    \put(60,85){\color{gray}$f^3_m$}
    \put(42.5,80){\color{gray}$f^3_n$}
    \put(76,10){\color{gray}$f^3_o$}
    \put(82,70){\color{UniBlue}{$\gamma_i$}}
  \end{overpic}
  \end{minipage}
  \caption{ A circular orthogonal trajectory $\gamma_i$ of
   a discrete cyclic system with three orthogonal surfaces
   $f^3_m, f^3_n$ and $f^3_o$ that pairwise form (up to
   orientation) a discrete Ribaucour pair. The cyan spheres
   provide the spheres enveloped by two orthogonal surfaces.}
\label{fig:scheme_cyclic}
\end{figure}
\noindent Algebraically, properties (i) and (ii) amount to
the following condition:
\begin{lem}\label{lem_orth_condition}
A discrete Legendre map $f$ is orthogonal to a circle
congruence $\Gamma$ if and only if any element $s_i \in f_i$
and $\mathfrak{c}_i \in \gamma_i^\perp \cap \mathcal{L}$
satisfy
\begin{equation}\label{equ_orth_legendre}
 \lspan{\mathfrak{s}_i, \mathfrak{c}_i
 + \lspan{\mathfrak{c}_i, \mathfrak{p}}\mathfrak{p}}=0 \ \ 
 \text{for any vertex } i \in \mathcal{V}.
\end{equation}
\end{lem}
\begin{proof}
Suppose that $s_i \in f_i \cap \PL$ is a point
sphere, then equation (\ref{equ_orth_legendre}) becomes
$\lspan{\mathfrak{s}_i, \mathfrak{c}_i}=0$. This holds if
and only if the point sphere $s_i$ lies on all spheres in
$\gamma_i^\perp$ and thus is a point on the circle.

For any other sphere $s_i \in f_i$, $\lspan{\mathfrak{s}_i,
\mathfrak{p}}\neq 0$, the claim follows from
(\ref{equ_orth_complex}).
\end{proof}

\noindent The circle congruences that we are particularly
interested in stem from discrete cyclic system:
\begin{defi}\label{def_cyclic_congruence}
A discrete circle congruence will be called \emph{cyclic} if
it admits a discrete family $I \ni \lambda \mapsto f^\lambda$
of orthogonal discrete Legendre maps that give rise to a
discrete cyclic system.
\end{defi}

\noindent In what follows, we will only consider
\emph{non-degenerate} circle congruences, that is, four
circles of any elementary quadrilateral do not share a
1-parameter family of  orthogonal spheres. By this assumption
we guarantee that not all orthogonal surfaces have totally
umbilic faces.
\begin{thm}\label{thm_flat_connection}
A non-degenerate circle congruence $\Gamma=(\gamma,
\gamma^\perp): \VV \rightarrow G_{(2,1)}^{\mathcal{P}}
\times G_{(2,1)}$ is cyclic if and only if it admits a flat
connection on the trivial bundle $\VV \times \mathbb{R}^{4,2}$
comprised of Lie inversions that map adjacent circles
onto each other, that is, there exists a 2-parameter family
of linear sphere complexes $a: \mathcal{E} \rightarrow
\mathbb{P}(\mathbb{R}^{4,2})$ such that the induced Lie
inversions $\{\sigma_{e}\}_{e\in \mathcal{E}}$ satisfy
\begin{equation*}
 \sigma_{ij}(\gamma_j) = \gamma_i \ \ \text{and} \ \ 
 \sigma_{ij} \circ \sigma_{jk} = \sigma_{il} \circ \sigma_{lk}
\end{equation*}
for any quadrilateral $(ijkl)$.
\end{thm}
\begin{rem} For degenerate circle congruences that admit a
flat connection of this kind the result can indeed fail:
consider, for example, a discrete circle congruence that
consists of parallel lines such that the point spheres
obtained as intersection of the lines with a fixed orthogonal
plane do not form a 2-dimensional circular net. Then, this
circle congruence does not give rise to a discrete cyclic
system, because there are no orthogonal surfaces that are
discrete Legendre maps. However, the reflections in the
bisecting planes of two adjacent parallel lines provide a flat
connection in the sense of Theorem \ref{thm_flat_connection}.
\end{rem} 

\begin{proof}
Suppose that the circle congruence $\Gamma: \VV \rightarrow
G_{(2,1)}^{\mathcal{P}} \times G_{(2,1)}$ is cyclic. Then,
by definition, there exists an associated discrete cyclic
system $f=(f^1,f^2,f^3)$ on $\VV \times I$, where $f^1$
and $f^2$ denote the two coordinate surface families of
discrete channel surfaces.

These discrete channel surfaces induce a canonical
M-Lie inversion between two adjacent circles $\Gamma_i$
and $\Gamma_j$ of $\Gamma$: namely, the M-Lie inversion
that interchanges the curvature circles of the constant Lie
cyclide provided by one of the discrete channel surfaces,
$f^1$ or $f^2$.
We denote these M-Lie inversions by $\sigma_{ij}$ and obtain:
$\sigma_{ij}(\gamma_j)=\gamma_i$.

Moreover, by construction, these M-Lie inversions map
adjacent contact elements (frames) $f_i$ and $f_j$
of the discrete cyclic system
onto each other. Since the orthogonal surfaces $f^3$ are
discrete Legendre maps, the connection provided by $\sigma$
is flat on $\gamma$. Furthermore, since, for any $i \in \VV$
and $m \in I$, we have $f_{i,m}^\mu \cap \gamma_i^\perp =:
s_i^\mu \in \LL$, $\mu=1,2$, it follows that
\begin{equation*}
 (\sigma_{ij} \circ \sigma_{jk} \circ \sigma_{kl}
  \circ \sigma_{li})(s_i^\mu) = s_i^\mu.
\end{equation*} 
Together with the fact that the M-Lie inversions $\sigma$
preserve the point sphere complex $\mathfrak{p}$, we conclude
that the induced connection is also flat on $\gamma^\perp$.

\truepar
Conversely, assume that a circle congruence $\Gamma$ admits
a flat connection provided by Lie inversions $\sigma_{ij}$
that map adjacent circles $\gamma_i$ and $\gamma_j$
onto each other, that is, any point sphere in $\gamma_i$
is mapped to a point sphere in $\gamma_j$. Thus the Lie
inversion $\sigma_{ij}$ preserves the point sphere complex
$\mathfrak{p}$ and is an M-Lie inversion.

To prove that $\Gamma$ is indeed cyclic, we will construct
a discrete family $I \ni \lambda \mapsto f^{3,\lambda}$
of orthogonal surfaces that is part of a discrete cyclic
system. To do so, we fix point spheres $I \ni \lambda
\mapsto p^\lambda_0$ along an initial circle $\Gamma_{i_0}$
of the circle congruence. Furthermore, according to Fact
\ref{fact_orth_cont_els}, we choose contact elements
$\lambda \mapsto f^{3,\lambda}_{i_0}$ that consist of
spheres orthogonal to $\Gamma_{i_0}$ such that any two
adjacent contact elements share a common sphere
(cf Figure \ref{fig:scheme_cyclic}).

Then transport of these contact elements $f_{i_0}^{3,\lambda}$
along edges $ij\in\mathcal{E}$ by the M-Lie inversions
$\sigma_{ij}$
consecutively defines a discrete family of orthogonal surfaces
$\lambda \mapsto f^{3,\lambda}$ by
\begin{equation*}
 f_j^{3,\lambda} := \sigma_{ji}(f_i^{3,\lambda}).
\end{equation*}
Due to the flatness of the connection, this construction
is well-defined. Moreover, since any M-Lie inversion
$\sigma_{ij}$ fixes a (curvature) sphere $s_{ij}^\lambda
\in f_i^{3,\lambda} \cap f_j^{3,\lambda}$ and the circle
congruence is non-degenerate, we indeed obtain discrete
Legendre maps $f^{3,\lambda}$.

\truepar
To complete the proof, it remains to construct two families
$\lambda_1 \mapsto f^{1,\lambda_1}$ and $\lambda_2 \mapsto
f^{2,\lambda_2}$ of discrete channel surfaces that supplement
the family $f^{3,\lambda}$ of orthogonal surfaces. To equip
the point sphere maps with suitable contact elements, we
choose two spheres $\mathfrak{r}_{i_0}^1, \mathfrak{r}_{i_0}^2
\in \gamma_{i_0}^\perp \cap \mathcal{L}$ that intersect
each other orthogonally and define contact elements along
the circle $\Gamma_{i_0}$ by
\begin{equation*}
 f^{\mu, i_0}_{\lambda} :
  = \lspann{\mathfrak{r}_{i_0}^\mu, p_{0}^\lambda} \ \ 
  \text{for } \ \mu \in \{ 1,2\}.
\end{equation*}
Then transport of these two contact elements by means of
the M-Lie inversions $\sigma$ along edges provides two
families of discrete channel surfaces,
with the circles of the congruence as curvature lines.

Thus, in summary, the constructed map $f:=(f^1, f^2, f^3)$
provides a discrete cyclic system with underlying circle
congruence $\Gamma$.
\end{proof}

\noindent Hence, to any discrete cyclic system we associate a
unique 2-parameter family of M-Lie inversions~$\sigma_{ij}$,
defined on edges of the associated circle congruence,
that simultaneously map adjacent point spheres of all
orthogonal surfaces onto each other. Note that these M-Lie
inversions are symmetric, that is, $\sigma_{ij}=\sigma_{ji}$.

\truepar
These M-Lie inversions induce a Ribaucour correspondence
between any two adjacent circles, which are therefore
cospherical. Furthermore, they reveal a property well-known
from smooth cyclic systems:
\begin{cor}\label{cor_const_cross}
The cross-ratio of the point spheres of any four orthogonal
surfaces of a discrete cyclic system is constant
along the orthogonal surfaces.
\end{cor}

\noindent Moreover, by using Fact \ref{fact_orth_cont_els},
we observe another relation between any two orthogonal
surfaces of a discrete cyclic system (see also Figure
\ref{fig:scheme_cyclic}):
\begin{cor}\label{cor_adjacent_rib}
Up to a possible change of orientation, any two orthogonal
surfaces of a discrete cyclic system are related by a discrete
Ribaucour transformation.
\end{cor}

\noindent However we note that, in contrast to the smooth
theory, the reconstruction of a discrete cyclic system
from a cyclic circle congruence is not unique. There are
several ambiguities: firstly, a discrete cyclic circle
congruence can admit various flat connections in the sense
of Theorem \ref{thm_flat_connection}. This non-uniqueness
of the connection stems from the generic existence of two
Ribaucour correspondences between two cospherical circles
(cf Lemma \ref{lem_rib_trafo_lie_inversion}). An example
of different flat connections associated to a face of a
discrete cyclic circle congruence is illustrated in Figure
\ref{fig:ambiguity}.

\begin{figure}[t]
\hspace{-1.5cm}\begin{minipage}{7.5cm}
  \includegraphics[scale=0.238]{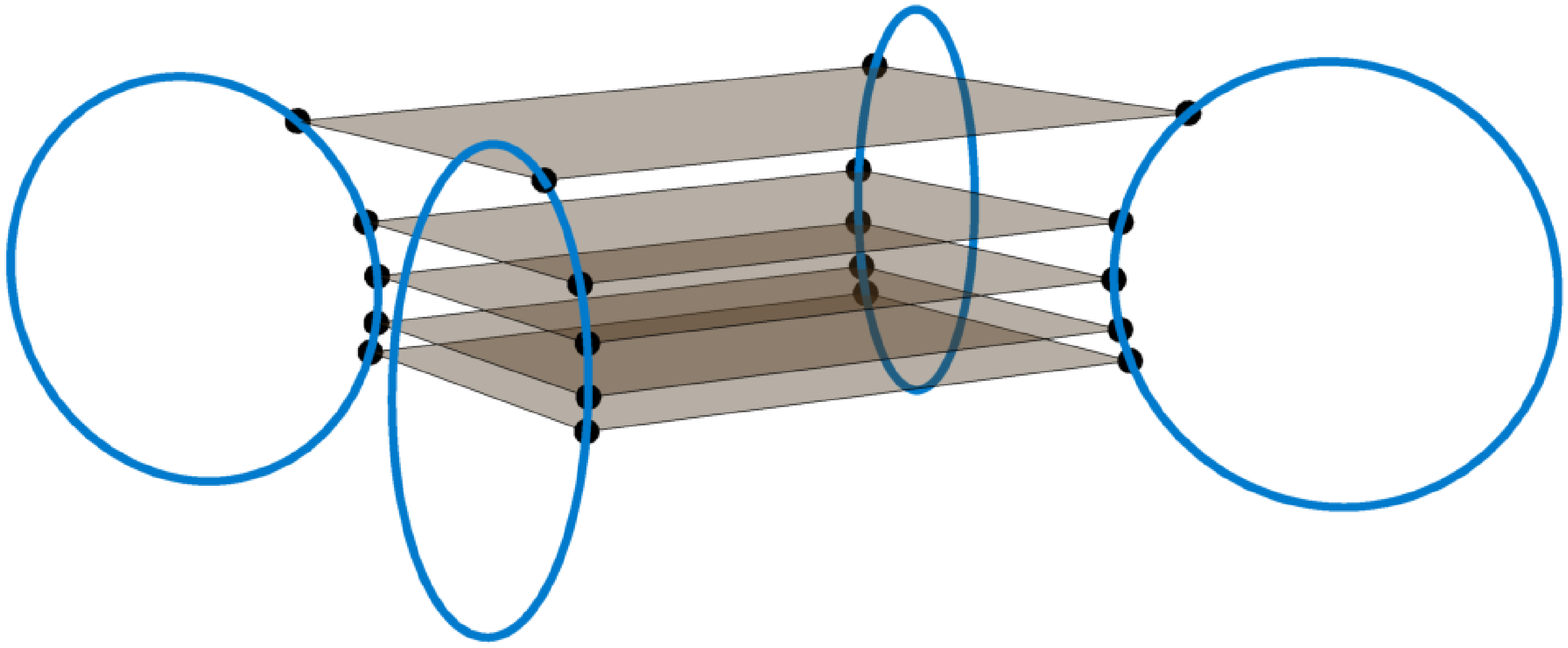}
 \end{minipage}
 \vspace*{-0.5cm} \begin{minipage}{6cm}
  \includegraphics[scale=0.25]{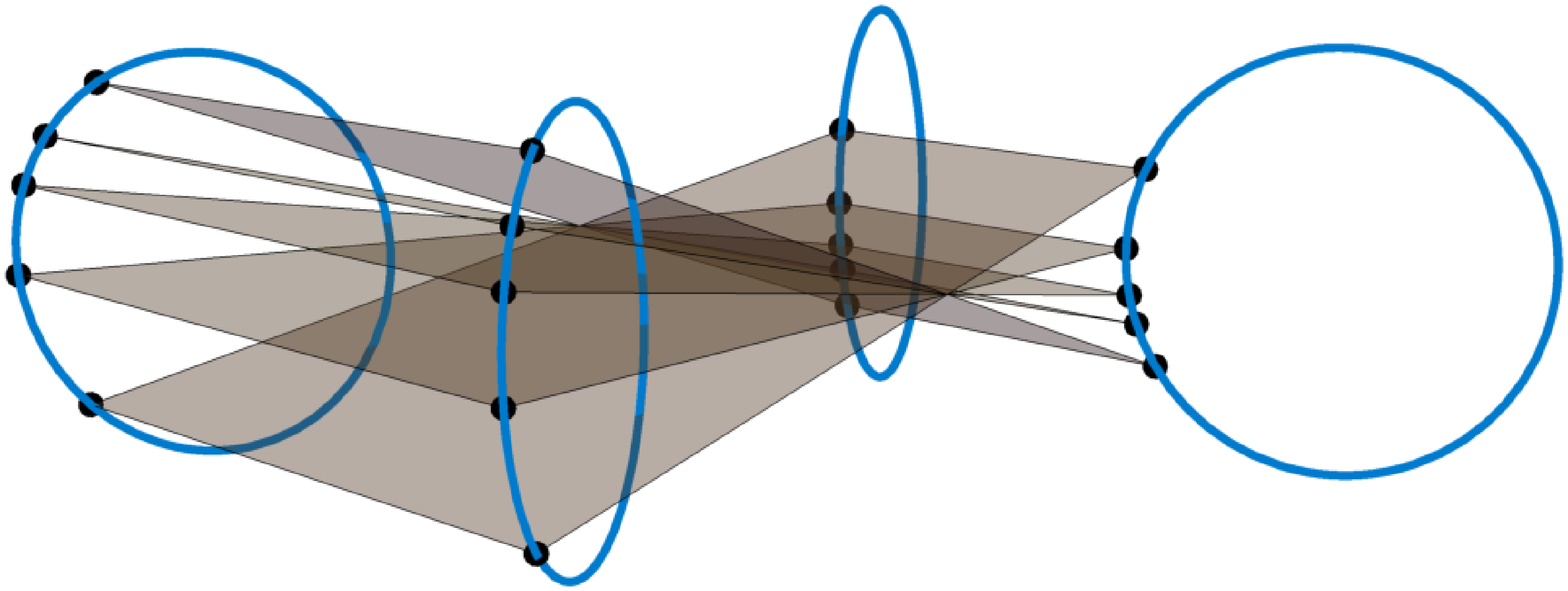}
 \end{minipage}
 \caption{Contrary to the smooth case, the family of
 orthogonal surfaces to a discrete cyclic circle congruence
 is not necessarily unique.}
\label{fig:ambiguity}
\end{figure}

Hence a discrete circle congruence does generically not have
a unique family of orthogonal surfaces,
even after an initial circle is equipped with a point sampling.

\truepar
Secondly, once a discrete family of orthogonal surfaces is
established, we are left with a $1$-parameter choice for
the other two families of coordinate surfaces, namely the
discrete channel surfaces. Although the circular curvature
lines are already fixed
we have the choice of a pair of orthogonally intersecting
curvature spheres for the discrete channel surfaces
at an initial circle
(see the choice of the spheres $r_{i_0}^1$ and $r_{i_0}^2$
in the proof of Theorem \ref{thm_flat_connection}).

\truepar
\noindent Recall \cite{MR0115134, cyclic_guichard_eth,
rib_cyclic} that, in the smooth case, the existence of three
orthogonal surfaces of a circle congruence already implies
that the congruence is cyclic.
We obtain an analogous statement in the discrete case.

We say that three orthogonal discrete Legendre maps are
\emph{generic} if, for any edge of the circle congruence,
the three curvature spheres span a $(2,1)$-plane of
$\mathbb{R}^{4,2}$.
\begin{thm}\label{thm_three_cyclic}
A discrete circle congruence that admits three generic
orthogonal discrete Legendre maps is cyclic.
\end{thm}
\begin{proof}
Let $\Gamma:\VV \rightarrow G_{(2,1)}^{\mathcal{P}} \times
G_{(2,1)}$ be a circle congruence that admits at least
three generic orthogonal discrete Legendre maps $g^\mu$,
$\mu \in \{1,2,3 \}$. To prove that $\Gamma$ is cyclic,
we will construct a flat connection as described in Theorem
\ref{thm_flat_connection}.

Firstly, consider two adjacent circles $\Gamma_i$ and
$\Gamma_j$ of the congruence $\Gamma$ and denote by
\begin{equation*}
 \mathfrak{s}_i, \mathfrak{t}_i \in \gamma_i^\perp \cap \LLL
\end{equation*}
two spheres that intersect in the circle $\Gamma_i$,
and similarly for $\Gamma_j$.
Since the curvature spheres $s_{ij}^\mu$ of the orthogonal
surfaces that belong to the edge $(ij)$ intersect all spheres
in $\gamma_i^\perp$ and $\gamma_j^\perp$ orthogonally
we obtain
\begin{equation}\label{equ_space}
 s_{ij}^\mu \in \lspann{\mathfrak{s}_1 + \lspan{\mathfrak{s}_1,
  \mathfrak{p}}\mathfrak{p}, \  
 \mathfrak{s}_2 + \lspan{\mathfrak{s}_2,
  \mathfrak{p}}\mathfrak{p}, \ 
 \mathfrak{t}_1 + \lspan{\mathfrak{t}_1,
  \mathfrak{p}}\mathfrak{p}, \ 
 \mathfrak{t}_2 + \lspan{\mathfrak{t}_2,
  \mathfrak{p}}\mathfrak{p}}^\perp.
\end{equation}
Hence we deduce that this span is at most 3-dimensional,
which implies that the space
\begin{equation*}
 \lspann{\mathfrak{s}_1, \mathfrak{s}_2, \mathfrak{t}_1,
 \mathfrak{t}_2, \mathfrak{p}} \subset \mathbb{R}^{4,2}
\end{equation*}
is at most 4-dimensional.
Therefore, due to Fact \ref{fact_circles_on_sphere},
the two circles $\Gamma_1$ and $\Gamma_2$ lie on a common
sphere.
This sphere is unique up to orientation.

As a consequence, using Corollary \ref{cor_choice_ribtrafo},
there exists a unique Ribaucour correspondence between any
two adjacent circles that extends the correspondence
induced by the point sphere maps of the orthogonal surfaces
$\{g^\mu\}_\mu$. By Lemma \ref{lem_rib_trafo_lie_inversion},
this correspondence is described by an M-Lie inversion
$\sigma_{ij}$.

Since $\{g^\mu\}_\mu$ are discrete Legendre maps, these M-Lie
inversions provide a flat connection for $\Gamma$.
Consequently, the circle congruence $\Gamma$ is indeed cyclic.
\end{proof}

To illustrate the situation we investigate a standard
example for cyclic circle congruences, namely, those that
stem from a parallel family of surfaces in a space form.
To start with we discuss discrete parallel surfaces in a
flat ambient space form as described in \cite{MR2657431}:
\begin{ex}\label{ex_parallel_flat}
Fix a point sphere complex $\mathfrak{p}$ and
a flat space form by choosing a lightlike space form vector $q
\in \LL$, $\lspan{\mathfrak{p}, \mathfrak{q}}=0$. Let $f:\VV
\rightarrow \mathcal{Z}$ be a discrete Legendre map with
space form projection
\begin{equation*}
 (\mathfrak{f}^\mathfrak{p}, \mathfrak{t}):
  \VV \rightarrow \QQ \times \PP,
\end{equation*}
where $\mathfrak{f}^\mathfrak{p}$ and $\mathfrak{t}$
denote the point sphere map and tangent plane congruence
in $f$, respectively.
Moreover, we denote the family of parallel discrete
Legendre maps of $f$ by $I \ni \lambda \rightarrow
(\mathfrak{f}^{\mathfrak{p},\lambda}, \mathfrak{t^\lambda})$.

At each vertex, the point spheres of the parallel family
lie on the line orthogonal to the tangent plane. Using Fact
\ref{fact_orth_circle}, this line is described by
\begin{equation*}
 \mathfrak{f}_i^{\mathfrak{p}, \lambda}
  \in \lspann{\mathfrak{f}_i^\mathfrak{p}, \mathfrak{q},
  \mathfrak{t}_i + \lspan{\mathfrak{t}_i,
  \mathfrak{p}}\mathfrak{p}} =: \gamma_i.
\end{equation*}
Moreover, note that the point spheres
$\mathfrak{f}_i^{\mathfrak{p}, \lambda}$ and
$\mathfrak{f}_j^{\mathfrak{p}, \lambda}$ of adjacent vertices,
as well as adjacent tangent planes, are related by the
reflection in the corresponding bisecting hyperplane.
Hence the sphere complexes
\begin{equation*}
 a:\mathcal{E} \rightarrow \PL, \ \mathfrak{a}_{ij} :
  = \mathfrak{t}_i - \mathfrak{t}_j
\end{equation*}
provide a flat connection for the discrete line congruence
$\Gamma=(\gamma, \gamma^\perp)$, which is therefore cyclic.

The two further coordinate surface families of a discrete
cyclic system associated to this line congruence consist of
developable discrete channel surfaces. They are determined
by two suitably prescribed curvature spheres at one initial
vertex of the discrete line congruence, that is, by the choice
of two orthogonal planes that intersect in the line of the
congruence at this initial vertex.
\end{ex}

\noindent To conclude this section, we emphasize
that our analysis also leads to a notion of
\emph{semi-discrete cyclic systems}, where the orthogonal
surfaces are discrete Legendre maps and the other two
coordinate surface families consist of semi-discrete
channel surfaces --- as has already become evident
from Example \ref{ex_parallel_flat}.

In particular, any discrete cyclic circle congruence gives
rise to semi-discrete cyclic systems: once a flat connection
in the sense of Theorem \ref{thm_flat_connection}
is established, we choose
a (smooth) 1-parameter family of orthogonal contact elements
 along an initial circle, as well as
two suitable curvature spheres for the semi-discrete channel
 surfaces that orthogonally intersect in this initial circle.
The transport of this initial principal frame by means of
the M-Lie inversions of the connection then gives
rise to the sought-after semi-discrete cyclic system.


\section{Discrete cyclic systems with special orthogonal surfaces}
\label{sect_from_rib}
\noindent
Historically, smooth cyclic systems were closely related to
the Ribaucour transformation of surfaces,
as described by Ribaucour \cite{rib_cyclic, salkowski}:
the circle congruence formed by the circles that orthogonally
intersect the surfaces of a Ribaucour pair in corresponding
points is cyclic;
furthermore, any two orthogonal surfaces of this cyclic
circle congruence form a Ribaucour pair.

By imposing special geometric properties on a Ribaucour pair
of surfaces,
this construction gives rise to particular cyclic systems.
Amongst them are
cyclic systems with orthogonal surfaces that are Guichard
 surfaces \cite{MR239516, csg_45} and,
in particular, parallel families of flat fronts in hyperbolic
 space~\cite{MR2652493};
and cyclidic systems, where all coordinate surfaces are
 Dupin cyclides \cite{darboux_ortho, totallycyclic}.

We report on a similar construction in the discrete framework:
the orthogonal circle congruence of a discrete Ribaucour pair
is cyclic.
Subsequently, we investigate circle congruences constructed
from discrete Ribaucour pairs with (at least) one totally
umbilic discrete envelope.
In this way, we obtain discrete cyclic circle
congruences with discrete flat fronts in hyperbolic space
as orthogonal surfaces, as well as discrete cyclic systems
where all coordinate surfaces are discrete Dupin cyclides.
Cyclic systems with a $1$-parameter family of orthogonal
Guichard surfaces will be reported on in a forthcoming paper.

\truepar
In this section we only consider \textit{non-degenerate} 
discrete Ribaucour sphere congruences 
$r:\mathcal{V} \rightarrow \mathbb{P}(\mathcal{L})$ 
(cf \cite{org_principles}),
that is, discrete conjugate nets (nets with planar faces)
in $\mathbb{P}(\mathcal{L})$,
where any homogeneous coordinate vectors of each face
span a $(2,1)$- or a $(1,2)$-plane. In the former case, 
where every span has Minkowski signature, each face 
of the sphere congruence models a Dupin cyclide,
that we will refer to as the associated \emph{R-Dupin cyclide}.

The notion of a
Ribaucour sphere congruence is clearly a Lie sphere geometric
notion. However, once we start to construct discrete cyclic 
circle congruences, we again fix a M\"obius subgeometry of Lie 
sphere geometry, modelled on
a point sphere complex $\mathfrak{p}$. 
We say that a discrete Ribaucour sphere congruence is 
\emph{non-degenerate in this M\"obius subgeometry} if it 
additionally satisfies 
$\lspan{\mathfrak{r}_i, \mathfrak{p}} \neq 0$,
that is, if it does not contain any point spheres.

\begin{figure}[t]
\hspace*{-3cm}\begin{minipage}{5cm}
  \includegraphics[scale=0.25]{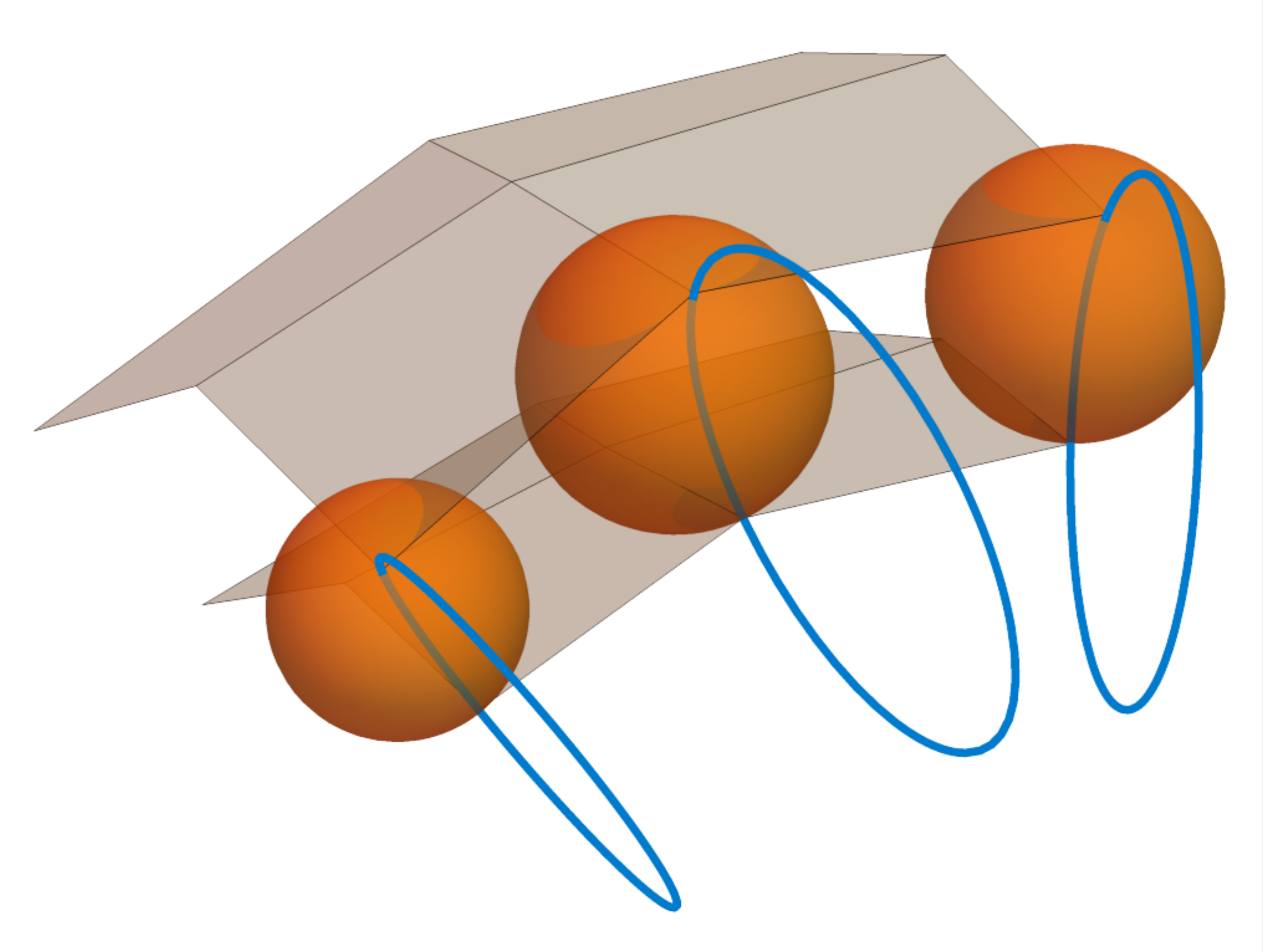}
 \end{minipage}
 \hspace*{1.2cm}\begin{minipage}{6cm}
  \includegraphics[scale=0.25]{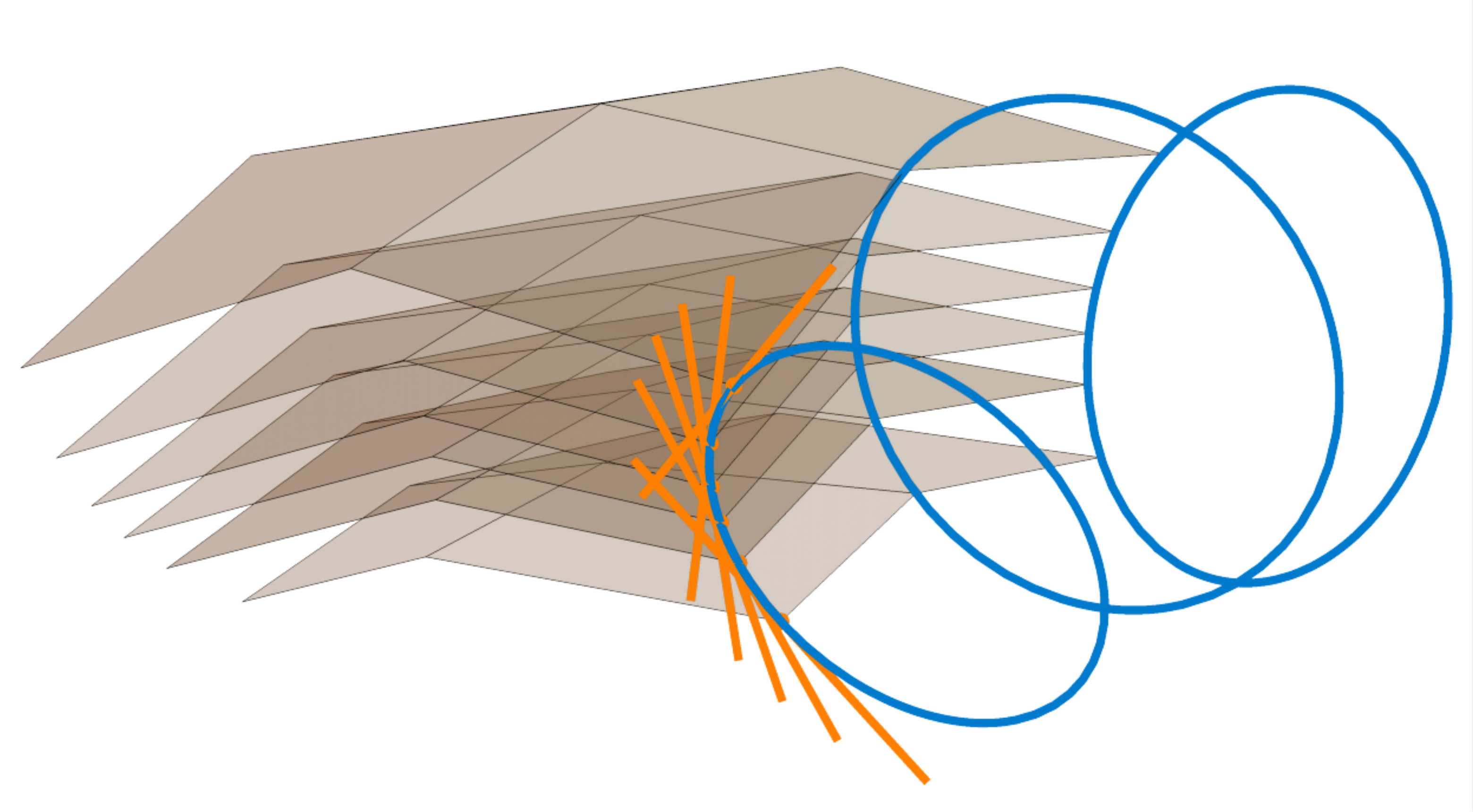}
 \end{minipage} 
 \caption{
  \emph{Left:} A Ribaucour pair of discrete Legendre maps
   with a part of its associated circle congruence (blue),
   consisting of circles that intersect the spheres of the
   Ribaucour congruence (orange) orthogonally in the point
   spheres of the envelopes.
  \emph{Right:} Discrete Legendre maps that are
   orthogonal to the circle congruence associated to the
   Ribaucour pair.}
\label{fig:ex_rib_pair}
\end{figure}

\subsection{Discrete cyclic circle congruences associated to
 discrete Ribaucour pairs}\label{subsect_rib}
We begin by formulating useful general observations about
Ribaucour sphere congruences:
first we provide a characterization in terms of the existence
of certain flat connections,
then show that a flat connection can be chosen to be comprised
of M-Lie inversions once a M\"obius subgeometry is fixed.
Given a pair of envelopes of the Ribaucour sphere congruence
this connection then also yields a flat connection for the
orthogonal circle congruence that we are interested in, 
showing that it is cyclic 
(cf Thm \ref{thm_flat_connection}). 

\truepar
Recall \cite{org_principles, rib_families} that any 
non-degenerate discrete Ribaucour sphere congruence 
admits a 2-parameter family of 
envelopes, which are uniquely determined by the choice of 
one appropriate contact element at one initial Ribaucour 
sphere.
 
Conversely, two discrete Legendre maps form a discrete 
Ribaucour pair if they envelop a common sphere congruence; 
this sphere congruence is then a discrete Ribaucour sphere 
congruence. 

\truepar
Taking the first point of view and focusing on the sphere 
congruence, we obtain the following characterization of 
discrete Ribaucour sphere congruences:

\begin{prop}\label{prop_ribsph_flat}
A discrete sphere congruence $r: \mathcal{V} \to \LL$ with
$\lspan{\mathfrak{r}_i, \mathfrak{r}_j} \neq 0$ for $(ij)
\in \mathcal{E}$, is a discrete Ribaucour sphere congruence
if and only if it admits a flat connection on the trivial
bundle $\mathcal{V} \times \mathbb{R}^{4,2}$ comprised of
Lie inversions that map adjacent Ribaucour spheres onto
each other.
\end{prop} 
\begin{proof}
Let $r: \mathcal{V} \to \LL$ be a discrete Ribaucour sphere
congruence and fix $n \in \mathbb{P}(\mathbb{R}^{4,2})$
such that $\lspan{\mathfrak{n}, \mathfrak{r}_i} \neq 0$
for any $i \in \VV$. We consider the discrete map
\begin{equation}\label{equ_rib_conn}
 a:\mathcal{E} \rightarrow \mathbb{P}(\mathbb{R}^{4,2}),
  \enspace 
 (ij)\mapsto a_{ij} := \lspann{\mathfrak{a}_{ij}},
  \enspace\text{where}\enspace
 \mathfrak{a}_{ij} := \lspan{\mathfrak{n}, \mathfrak{r}_i}\mathfrak{r}_j - \lspan{\mathfrak{n}, \mathfrak{r}_j}\mathfrak{r}_i
\end{equation}
and denote the Lie inversions with respect to the linear
sphere complex $a_{ij}$ by $\sigma_{ij}$. By construction,
those Lie inversions exchange adjacent spheres of the congruence
and satisfy, for any face $(ijkl)$,
\begin{equation*}
(\sigma_{ij}\circ\sigma_{jk}\circ\sigma_{kl}\circ\sigma_{li})
(\mathfrak{r}_i) = \mathfrak{r}_i.
\end{equation*}
To prove that the Lie inversions $\sigma$ indeed yield a
flat connection on  $\mathcal{V} \times \mathbb{R}^{4,2}$,
we will investigate the interplay between the 2-parameter 
family of envelopes of $r$ and these Lie inversions. 
Firstly, note that $\sigma_{ij}$
preserves all spheres that are in oriented contact with $r_i$
and $r_j$, hence, adjacent contact elements of the envelopes
are exchanged (see also \cite[\S 3]{rib_families}). 
In particular, we obtain that 
$(\sigma_{ij}\circ\sigma_{jk}\circ\sigma_{kl}\circ\sigma_{li})
(f_i) = f_i $ for the contact elements of any envelope $f$.

Moreover, since $\lspan{\mathfrak{a}_{ij},
\mathfrak{n}}=0$, the Lie inversions preserve the
linear sphere complex $\mathbb{P}(\mathcal{L} \cap n^\perp)$.
Since each contact element of an envelope
contains exactly one sphere in $\mathbb{P}(\mathcal{L}
\cap n^\perp)$ and any discrete Ribaucour sphere
congruence admits a 2-parameter family of discrete envelopes,
 we conclude that, for any $\mathfrak{v} \in r_i^\perp$,
\begin{equation*}
 (\sigma_{ij}\circ\sigma_{jk}\circ\sigma_{kl}\circ\sigma_{li})
 (\mathfrak{v})=\mathfrak{v}.
\end{equation*}
Thus, the Lie inversions $\sigma_{ij}$ provide
a flat connection for the Ribaucour sphere congruence.  

\truepar
Conversely, assume that $r:\mathcal{V} \to \LL$,
$\lspan{\mathfrak{r}_i, \mathfrak{r}_j} \neq 0$, is a
discrete sphere congruence that admits a flat connection
comprised of Lie inversions, as above.
Furthermore, let $s_0 \in \LL$ be a
sphere in oriented contact to an initial sphere
$r_0$ of the Ribaucour sphere congruence,
that is, $s_0 \perp r_0$. Then propagation of the
contact element $\lspann{s_0, r_0}$ with by means of the
flat connection provides a well-defined discrete Legendre
map that envelopes~$r$. Due to the possible choices for $s_0$,
we
obtain a 2-parameter family of discrete envelopes and $r$
is indeed a discrete Ribaucour sphere congruence.
\end{proof}

\noindent We remark that the flat connection of Prop
\ref{prop_ribsph_flat} is not unique. However, if
we fix a M\"obius subgeometry by choosing a point sphere
complex $\mathfrak{p}$ such that the Ribaucour sphere
congruence is non-degenerate in this M\"obius subgeometry, 
we may fix
homogeneous coordinates
$\mathfrak{r}:\mathcal{V}\to\mathcal{L}$ so that
$\lspan{\mathfrak{r},\mathfrak{p}}\equiv 1$ and consider
the discrete map
\begin{equation}\label{equ_a_rib_pair}
 a:\mathcal{E} \rightarrow \mathbb{P}(\mathbb{R}^{4,2}),
  \enspace 
 (ij)\mapsto a_{ij} := \lspann{\mathfrak{a}_{ij}},
  \enspace\text{where}\enspace
 \mathfrak{a}_{ij} := \mathfrak{r}_j - \mathfrak{r}_i.
\end{equation}
Then the associated Lie inversions $\sigma_{ij}$ preserve
the point sphere complex and are therefore M-Lie inversions
(cf \cite[\S 3.1]{rib_families}). Hence:

\begin{cor}\label{cor_rib_flat_Mlie}
A discrete sphere congruence $r: \mathcal{V} \to \LL
\setminus \PL$ with $\lspan{\mathfrak{r}_i, \mathfrak{r}_j}
\neq 0$ for $(ij) \in \mathcal{E}$, is a discrete Ribaucour
sphere congruence if and only if it admits a (unique)
flat connection on the trivial bundle $\mathcal{V} \times
\mathbb{R}^{4,2}$ compound by M-Lie inversions that map
adjacent Ribaucour spheres onto each other.
\end{cor}

\noindent Returning to our principal aim,
the construction of a cyclic circle
congruence associated to a discrete Ribaucour pair, 
we fix two discrete Legendre maps $(f^+,
f^-):\mathcal{V} \rightarrow \mathcal{Z}\times \mathcal{Z}$
 that envelop the non-degenerate
discrete Ribaucour sphere congruence $r:\mathcal{V} \rightarrow
\mathbb{P}(\mathcal{L})\setminus \PL$. 

Then the flat connection of M-Lie inversions associated to
$r$ (see Cor \ref{cor_rib_flat_Mlie}) will
also provide a flat connection for the circle congruence:
we consider the
circle congruence $\Gamma:\mathcal{V}
\rightarrow G_{(2,1)}^{\mathcal{P}} \times G_{(2,1)}$
consisting of circles $\Gamma_i$ that intersect the
spheres $r_i$ in the point spheres $p^\pm_i \in f^\pm_i$
of the envelopes orthogonally
(for an illustration see Figure~\ref{fig:ex_rib_pair}).
By Fact~\ref{fact_orth_circle}, this circle congruence is
described by the $(2,1)$-planes
\begin{equation}\label{equ_orth_circle_rib}
 \gamma_i:=\lspann{\mathfrak{p}_i^+, \
 \mathfrak{p}_i^-, \
 \mathfrak{r}_i+\mathfrak{p}}.
\end{equation}
Since the M-Lie inversions $\sigma$ described by
(\ref{equ_a_rib_pair}) satisfy $\sigma_{ij}(\mathfrak{r}_j +
\mathfrak{p})=\mathfrak{r}_i+ \mathfrak{p}$, they also map
adjacent circles of $\Gamma$
onto each other, that is, $\sigma_{ij}(\gamma_j)=\gamma_i$. 

Thus, by the above, these M-Lie inversions $\sigma$ yield a flat
connection for the circle congruence $\Gamma$ and, by
Theorem~\ref{thm_flat_connection}, we conclude:

\begin{thm}\label{thm_associated_rib_cyclic}
Let $p^\pm \in f^\pm$ be the point sphere maps of two
envelopes $f^\pm$ of a discrete Ribaucour sphere congruence
$r:\VV \rightarrow \LL$.
Then the circles that orthogonally intersect the spheres
$r$ in the point sphere maps $p^\pm$ form a cyclic circle
congruence with $f^\pm$ as orthogonal surfaces.
This circle congruence is given by
\begin{equation*}
\begin{aligned}
 \Gamma=(\gamma, \gamma^\perp):
  \VV &\rightarrow G_{(2,1)}^{\mathcal{P}} \times G_{(2,1)},
 \\i &\mapsto \lspann{\mathfrak{p}_i^+, \mathfrak{p}_i^-,
 \mathfrak{r}_i+\mathfrak{p}} \times \lspann{\mathfrak{p}_i^+,
 \mathfrak{p}_i^-, \mathfrak{r}_i+\mathfrak{p}}^\perp,
\end{aligned}
\end{equation*}
where $\mathfrak{r}_i \in r_i$ such that
$\lspan{\mathfrak{r}_i, \mathfrak{p}}=1$;
it will be referred to as \emph{associated}
to the discrete Ribaucour pair.
\end{thm}

\noindent Firstly, we will exploit this construction to
extend Example \ref{ex_parallel_flat} and discuss discrete
(normal) line congruences in space forms obtained from parallel
discrete surfaces.
\begin{ex}[Discrete cyclic circle congruences associated
 to parallel surfaces in space forms]
Let $\mathfrak{p} \in \mathbb{R}^{4,2}$, $\lspan{\mathfrak{p},
\mathfrak{p}}=-1$, be a fixed point sphere complex and
$\mathfrak{q} \in \mathbb{R}^{4,2}\setminus \{ 0 \}$,
$\lspan{\mathfrak{p}, \mathfrak{q}}=0$, a space form
vector satisfying $\lspan{\mathfrak{q},\mathfrak{q}}=\pm 1$.
As before, we denote the space form projection of
a discrete Legendre map $f$ by
\begin{equation*}
 (\mathfrak{f}^\mathfrak{p}, \mathfrak{t}):
  \VV \rightarrow \QQ \times \PP,
\end{equation*}
with the point sphere map $\mathfrak{f}^\mathfrak{p}$ and
tangent plane congruence $\mathfrak{t}$ of $f$, respectively.
Furthermore, consider its discrete Ribaucour transform
$\hat{f}:=\sigma_q(f)$ obtained by the Lie inversion
in the linear sphere complex $\mathbb{P}(\LLL \cap \{ \mathfrak{q}\}^\perp)$.
 
By Theorem \ref{thm_associated_rib_cyclic}, the
discrete cyclic circle congruence associated to the Ribaucour
pair $(f, \hat{f})$ is given by
\begin{equation*}
 \gamma_i :
  = \lspann{\mathfrak{f}_i^p, \hat{\mathfrak{f}}_i^p,
  \mathfrak{t}_i+\lspan{\mathfrak{t}_i,
  \mathfrak{p}}\mathfrak{p}} = \lspann{\mathfrak{f}_i^p,
  \mathfrak{q}, \mathfrak{t}_i-\mathfrak{p}},
\end{equation*}
which yields a discrete normal line congruence in the
distinguished space form, that admits the family of parallel
surfaces as its orthogonal surfaces.

\truepar
In the case of a hyperbolic ambient quadric of constant
curvature,
that is, $\lspan{\mathfrak{q}, \mathfrak{q}} > 0$, the circles
of the constructed congruence $\Gamma$ orthogonally intersect
the spheres $\mathfrak{l}^\pm = \mathfrak{p} \pm \mathfrak{q}$
that coincide up to orientation.
Those spheres $\mathfrak{l}^\pm$ represent the infinity
boundary of the hyperbolic quadric of constant curvature,
that consists of two hyperbolic space forms.
Moreover, the
two orthogonal surfaces whose point spheres lie on $l^\pm$
provide the two discrete hyperbolic Gauss maps of $f$ and
its parallel surfaces.

Thus, by Corollary \ref{cor_adjacent_rib}, we learn that
the two discrete hyperbolic Gauss maps of a discrete Legendre
map are (up to orientation) related by a discrete Ribaucour
transformation. This fact suggests another construction for
discrete cyclic circle congruences associated to parallel
surfaces in hyperbolic space, namely, from its two
discrete hyperbolic Gauss maps.

\truepar
We say that a discrete Legendre map is \emph{totally umbilic}
if all curvature spheres coincide and, for any choice of
point sphere complex, its point sphere map is a circular net.

Using this definition, we have proven:
\begin{cor}\label{cor_rib_hyper_gauss}
A discrete Ribaucour pair of two totally umbilic discrete
Legendre maps whose point spheres lie on the same sphere gives
rise to a discrete normal line congruence in an appropriate
hyperbolic space form.
Any of its orthogonal surfaces in this hyperbolic space form
are parallel surfaces,
and have the discrete
Ribaucour pair as their common discrete hyperbolic Gauss maps.
\end{cor}
\end{ex}

\noindent Clearly, the geometry of the chosen discrete
hyperbolic Gauss maps affects the properties of the
associated orthogonal surfaces. For example, in the next
subsection, we shall see that starting with hyperbolic Gauss
maps that form a Darboux pair will lead to discrete flat
fronts in hyperbolic space.

\subsection{Discrete flat fronts as orthogonal surfaces
 of cyclic circle congruences} \label{subsect_flat_fronts}
\noindent In the smooth, as well as in the discrete setup,
there are various ways to construct intrinsically flat
surfaces (or fronts) in hyperbolic space, that
is, surfaces with constant extrinsic Gaussian curvature~$1$.

From a Lie sphere geometric perspective, smooth and discrete
flat fronts in hyperbolic space are obtained as projections of
$\Omega$-surfaces spanned by two isothermic sphere congruences
each lying in a fixed parabolic linear sphere complex (see
\cite{MR2652493, lin_weingarten}).
Alternatively, as discussed in \cite{MR1752778, MR2946924,
disc_sing}, smooth and discrete flat fronts can be
produced from holomorphic data by means of a
Weierstrass type representation.

In \cite{MR2652493}, it was shown that smooth flat fronts
also arise as orthogonal surfaces of special cyclic systems
associated to Darboux pairs of totally umbilic surfaces,
namely, of their two hyperbolic Gauss maps. The aim of this
subsection is to demonstrate a similar construction within
the framework developed here, that also leads to parallel
families of discrete flat fronts in hyperbolic space.

\truepar
In order to make contact with our present setting,
we briefly recall the Lie geometric approach to discrete
flat fronts in hyperbolic space as established in
\cite[Expl~4.3]{lin_weingarten}:
a discrete $\Omega$-surface is a Legendre map $f=s^+\oplus s^-$
that is spanned by a suitable pair of isothermic sphere
congruences,
and its hyperbolic space form projection
$(\mathfrak{f}^\mathfrak{p},\mathfrak{t}):
 \mathcal{V}\to\QQ\times\PP$
with respect to the point sphere complex~$\mathfrak{p}$
and space form vector $\mathfrak{q}$ is a flat front
if and only if the enveloped isothermic sphere congruences
take values in two fixed parabolic linear sphere complexes
$l^\pm=\lspann{\mathfrak{q}\mp\mathfrak{p}}$,
that is,
$s^\pm\perp l^\pm$.

Fixing homogeneous coordinate vectors
$\mathfrak{l}^\pm=\mathfrak{q}\mp\mathfrak{p}$,
K\"onigs dual lifts of the two sphere congruences $s^\pm$
are determined by
\begin{align*}\label{equ_def_s}
 \mathfrak{s}^\pm
 = \pm( \lspan{\mathfrak{l}^\pm,\mathfrak{t}}
         \mathfrak{f}^\mathfrak{p}
      - \lspan{\mathfrak{l}^\pm,\mathfrak{f}^\mathfrak{p}}
         \mathfrak{t} )
 = \mathfrak{f}^\mathfrak{p} \pm \mathfrak{t};
\end{align*}
that is, they are edge-parallel and opposite diagonals on
each face are parallel:
\begin{equation*}
 \s_i^+-\s_j^+ \parallel \s_i^--\s_j^-, \enspace
 \s_i^+-\s_k^+ \parallel \s_j^--\s_l^- 
  \enspace\text{and}\enspace
 \s_j^+-\s_l^+ \parallel \s_k^--\s_i^-.
\end{equation*}
The former condition is equivalently expressed by
the vanishing of the mixed area on faces,
$$
  0 = A(\mathfrak{s}^+,\mathfrak{s}^-),
   \enspace\text{while}\enspace
  A(\mathfrak{s}^+,\mathfrak{s}^-)
    = A(\mathfrak{f}^\mathfrak{p},\mathfrak{f}^\mathfrak{p})
    - A(\mathfrak{t},\mathfrak{t}),
$$
showing that the mixed area Gauss curvature
$K:=\frac{A(\mathfrak{t},\mathfrak{t})}
 {A(\mathfrak{f}^\mathfrak{p},\mathfrak{f}^\mathfrak{p})}$
of $(\mathfrak{f}^\mathfrak{p},\mathfrak{t})$ satisfies
$K\equiv 1$ if and only if
$\mathfrak{s}^\pm$ are K\"onigs dual.

\truepar
The key point in the construction of discrete cyclic circle
congruences that admit a parallel family of discrete flat
fronts is the interplay between the distinguished isothermic
sphere congruences of the discrete flat fronts and their
hyperbolic Gauss maps as totally umbilic envelopes of them.

The proposed construction will rely on the following general
observations on discrete Ribaucour sphere congruences with
a totally umbilic (discrete) envelope:
\begin{lem}[cf \cite{rib_families}]
 \label{lem_umbilic_envelope}
Let $r$ be a discrete Ribaucour sphere congruence that admits
a totally umbilic envelope, then $r$ is a $(2,1)$-congruence
and, on any face, the contact elements of the totally umbilic
envelope coincide with contact elements of the corresponding
R-Dupin cyclide along one circular curvature line of it.
\end{lem}

\noindent In this situation, the cross-ratios of the discrete
Ribaucour congruence are transferred to the point sphere map
of the totally umbilic envelope and vice versa. This follows
from a simple fact about smooth Dupin cyclides:
\begin{lem}\label{lem_same_cr}
The cross-ratio of four point spheres lying on a curvature
line of a Dupin cyclide coincides with the cross-ratio of the
four (non-constant) curvature spheres of the Dupin cyclide
that are in oriented contact with those point spheres.
\end{lem}
\begin{proof}
Let $\Delta=\delta_1 \oplus_\perp \delta_2 \in G_{(2,1)}
\times G_{(2,1)}$ represent a Dupin cyclide;
further
let $\mathfrak{s}_1\in\delta_1\cap\LLL$
and $\mathfrak{s}_{2j}\in\delta_2\cap\LLL$ ($j=1,\dots,4$)
denote one, respectively four, curvature spheres
of different families,
that is, $f_j=s_1\oplus s_{2j}$ ($j=1,\dots,4$) yield
four contact elements of the cyclide.
The point spheres of the corresponding contact elements
are then given by
$$
  \mathfrak{p}_j
  = \lspan{\mathfrak{p},\mathfrak{s}_1}\mathfrak{s}_{2j}
  - \lspan{\mathfrak{p},\mathfrak{s}_{2j}}\mathfrak{s}_1.
$$
As
$\lspan{\mathfrak{p}_i,\mathfrak{p}_j}
 =\lspan{\mathfrak{s}_{2i},\mathfrak{s}_{2j}}$
we conclude that
$cr(p_1,p_2,p_3,p_4)=cr(s_{21},s_{22},s_{23},s_{24})$.
\end{proof}

\noindent As a consequence of Lemmas
\ref{lem_umbilic_envelope} and \ref{lem_same_cr} we then
obtain:
\begin{cor}\label{cor_umbilic_cr}
The face cross-ratios
of a discrete Ribaucour sphere congruence and
of a totally umbilic envelope coincide.
In particular, a totally umbilic envelope of a discrete
Ribaucour sphere congruence is isothermic if and only if
the sphere congruence is.
\end{cor}

\noindent In view of the fact that a discrete flat front is
spanned by a pair of isothermic sphere congruences, and those
are discrete Ribaucour sphere congruences that each have one
of the discrete hyperbolic Gauss maps as totally umbilic
second envelopes, we state the following theorem
(see also Figure \ref{fig:flat_fronts}):
\begin{thm}\label{thm_flat_front_final}
The hyperbolic Gauss maps $h^\pm:\VV\to\PL$ of a discrete
flat front in hyperbolic space form a (totally umbilic)
Darboux pair.\par
Conversely, any orthogonal net of the cyclic circle
congruence associated to a Darboux pair $(h^+,h^-)$
with values in a $2$-sphere projects to a flat front in
the hyperbolic space bounded by the target sphere of $h^\pm$.
\end{thm}

\begin{proof}
First suppose $(\mathfrak{f}^\mathfrak{p},\mathfrak{t})$ to be
a flat front in a hyperbolic space described by a point sphere
complex $\mathfrak{p}$ and a space form vector $\mathfrak{q}$
with $\lspan{\mathfrak{q},\mathfrak{q}}=1$.
Using homogeneous coordinate vectors as above,
$\mathfrak{s}^\pm=\mathfrak{f}^\mathfrak{p}\pm\mathfrak{t}$ and
$\mathfrak{l}^\pm=\mathfrak{q}\mp\mathfrak{p}$ for
the enveloped isothermic sphere congruences $s^\pm$ and
their linear sphere complexes, respectively,
we obtain
$$
 \mathfrak{h}^\pm
 = \pm( \lspan{\mathfrak{p},\mathfrak{l}^\pm}\mathfrak{s}^\pm
      - \lspan{\mathfrak{p},\mathfrak{s}^\pm}\mathfrak{l}^\pm )
 = \mathfrak{s}^\pm + \mathfrak{l}^\pm
 = (\mathfrak{f}^\mathfrak{p}+\mathfrak{q}) \pm (\mathfrak{t}-\mathfrak{p})
$$
as (homogeneous coordinates of the) hyperbolic Gauss maps
of the front $(\mathfrak{f}^\mathfrak{p},\mathfrak{t})$.
Note that $\mathfrak{h}^\pm\perp\mathfrak{p},\mathfrak{q}$,
that is, they are point sphere maps taking values in the
sphere $\mathfrak{q}$ that defines the infinity boundary
of the ambient hyperbolic space.
As $\mathfrak{h}^\pm$ are just constant offsets
of $\mathfrak{s}^\pm$,
they share cross ratios with the enveloped isothermic sphere
congruences,
hence are also isothermic (cf Cor \ref{cor_umbilic_cr}),
and form a pair of K\"onigs dual lifts,
hence yield a totally umbilic Darboux pair $(h^+,h^-)$,
see \cite[Def~4.4, Thm~3.26]{bcjpr}.
Note that $\lspan{\mathfrak{h}^+,\mathfrak{h}^-}=-2\neq 0$,
showing that the obtained Darboux pair is non-isotropic.

Next we reverse the above construction:
thus let $(h^+,h^-)$ denote a discrete Darboux pair of point
spheres that take values in a $2$-sphere,
that is, $h^\pm\perp\mathfrak{p},\mathfrak{q}$,
where $\mathfrak{p}$ denotes the point sphere complex
and $\mathfrak{q}\in\mathbb{R}^{4,2}$ with
$\lspan{\mathfrak{q},\mathfrak{q}}=1$
defines a M\"obius geometric sphere.
Then $\mathfrak{l}^\pm:=\mathfrak{q}\mp\mathfrak{p}$ yield
the two (Lie geometric) oriented spheres of this
(M\"obius geometric) sphere.
By \cite[Def~4.4, Thm~3.26]{bcjpr} we may choose K\"onigs dual
lifts $\mathfrak{h}^\pm$ of $h^\pm$;
as $\lspan{\mathfrak{h}^\pm,\mathfrak{h}^\pm}\equiv 0$
we infer (as usual) that
$$
 \mathfrak{h}^\pm_i + \mathfrak{h}^\pm_j
 \perp \mathfrak{h}^\pm_i - \mathfrak{h}^\pm_j
  \parallel \mathfrak{h}^\mp_i - \mathfrak{h}^\mp_j,
  \enspace\text{hence}\enspace
 \lspan{\mathfrak{h}^+,\mathfrak{h}^-} \equiv -2
$$
without loss of generality, after a possible (constant)
rescaling of $\mathfrak{h}^+$ or $\mathfrak{h}^-$.
Now observe that
$$
 \mathfrak{s}^\pm :
 = e^{\pm\rho}\mathfrak{h}^\pm + \mathfrak{l}^\pm
  \enspace\text{for}\enspace
 \rho\in\mathbb{R}
$$
yield K\"onigs dual lifts of an isotropic Darboux pair
$(s^+,s^-)$ of isothermic sphere congruences,
see \cite[Def 4.4]{bcjpr}:
in particular, the Legendre map $f:=s^+\oplus s^-$ projects
to a flat front $(\mathfrak{f}^\mathfrak{p},\mathfrak{t})$
in the hyperbolic space(s) given by $\mathfrak{p}$ and
$\mathfrak{q}$.

Finally note that, by Theorem \ref{thm_associated_rib_cyclic},
the cyclic circle congruence associated to the Darboux pair
$(h^+,h^-)$ is given by
$$
 \Gamma = \gamma\times\gamma^\perp:
 \mathcal{V}\to G_{(2,1)}^{\mathcal{P}} \times G_{(2,1)}
  \enspace\text{with}\enspace
 \gamma_i :
 = \lspann{ \mathfrak{h}^+_i,\mathfrak{h}^-_i,\mathfrak{q} }.
$$
Clearly, $\mathfrak{s}^\pm_i\pm\mathfrak{p}\in\gamma_i$
at every vertex $i\in\mathcal{V}$;
consequently, $f$ is an orthogonal net of this cyclic
circle congruence for every $\rho\in\mathbb{R}$,
by Fact \ref{fact_two_orth}.
\end{proof}

\noindent We note that the Ribaucour sphere congruence
enveloped by the discrete Darboux pair $(h^+,h^-)$,
considered in Theorem~\ref{thm_flat_front_final},
is highly degenerate:
it is the constant sphere $l^+$ or $l^-$ representing
the infinity boundary of hyperbolic space.
Hence the description (\ref{equ_a_rib_pair}) of the linear
sphere complexes that subsequently induce the flat connection
for the circle congruence fails.
However, in this particular situation the K\"onigs dual lifts
of the Darboux pair give rise to the sought-after linear sphere
complexes by
\begin{equation*}
 a_{ij}
 = \lspann{ \mathfrak{h}^+_i-\mathfrak{h}^+_j }
 = \lspann{ \mathfrak{h}^-_i-\mathfrak{h}^-_j }.
\end{equation*} 
These then induce the M-Lie inversions that provide the flat
connection for $\Gamma$ and interchange the point spheres of
the orthogonal nets on adjacent circles of the congruence.

\begin{figure}[t]
\hspace*{-4.7cm}\begin{minipage}{5cm}
  \begin{overpic}[scale=.4]{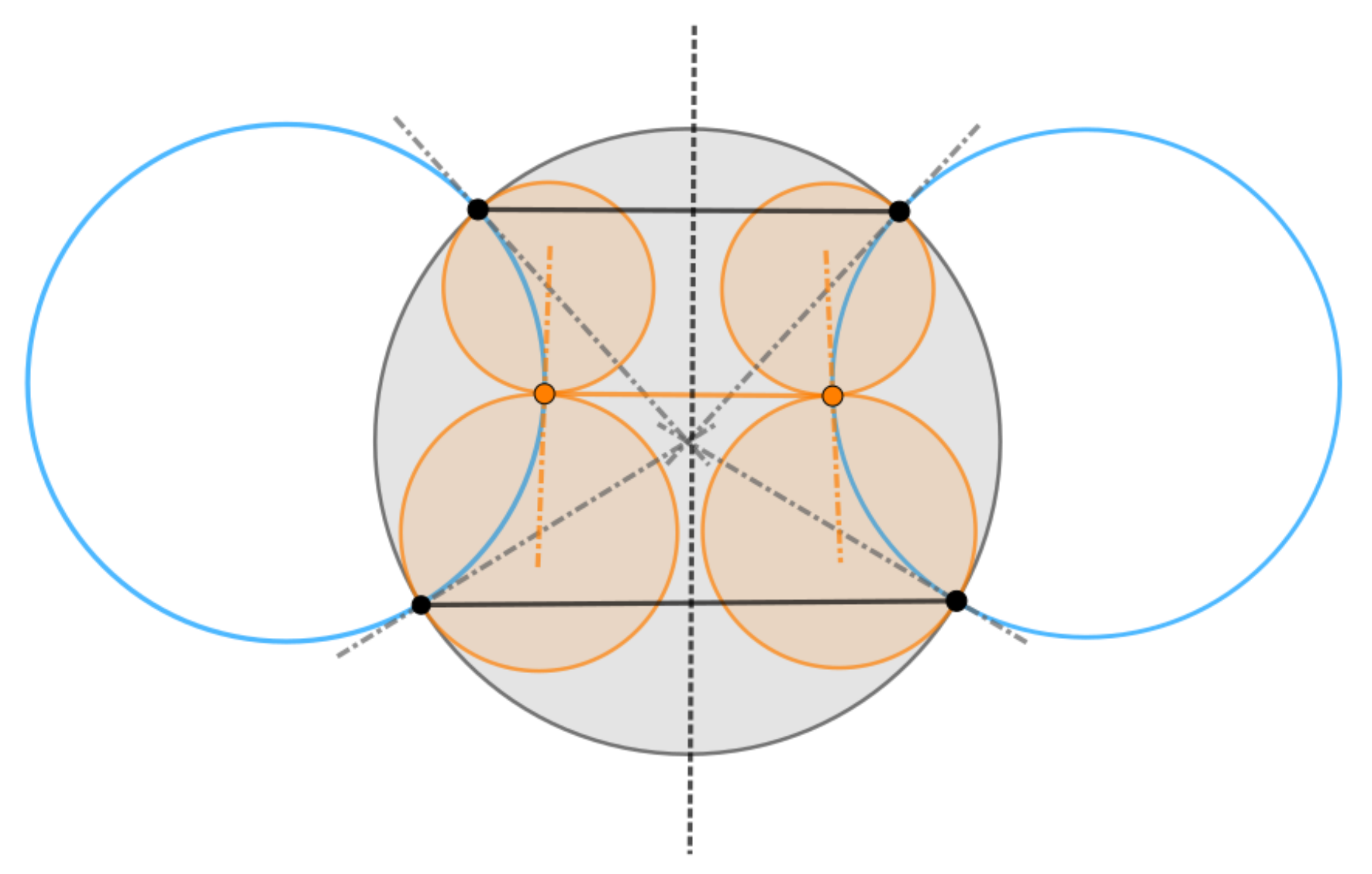}
  	\put(28,15){\color{gray}$h_i^-$}
    \put(68,15){\color{gray}$h_j^-$}
    \put(33,51){\color{gray}$h_i^+$}
    \put(63,51){\color{gray}$h_j^+$}
    \put(58,6){\color{gray}$l^\pm$}
    \put(52,58){\color{black}$\sigma_{ij}$}
    \put(39,22){\color{orange}{$f_i$}}
    \put(57.5,22){\color{orange}{$f_j$}}
    \put(34,39){\color{orange}{$s_i^+$}}
    \put(63,39){\color{orange}{$s_j^+$}}
    \put(30.5,27){\color{orange}{$s_i^-$}}
    \put(65.5,27){\color{orange}{$s_j^-$}}
    \put(3,38){\color{UniBlue}{$\gamma_i$}}
    \put(93,38){\color{UniBlue}{$\gamma_j$}}
   \end{overpic}
\end{minipage}

  \caption{Construction of a discrete flat front
   $f=\lspann{s^+, s^-}$ from the discrete cyclic circle
   congruence $\Gamma$ associated to a discrete Darboux pair
   of totally umbilic point sphere maps $h^\pm$
   with values in
   the fixed sphere $l^\pm$.
   Those become the point sphere maps of the two hyperbolic
   Gauss maps of the discrete flat front.}
  \label{fig:flat_fronts}
\end{figure}
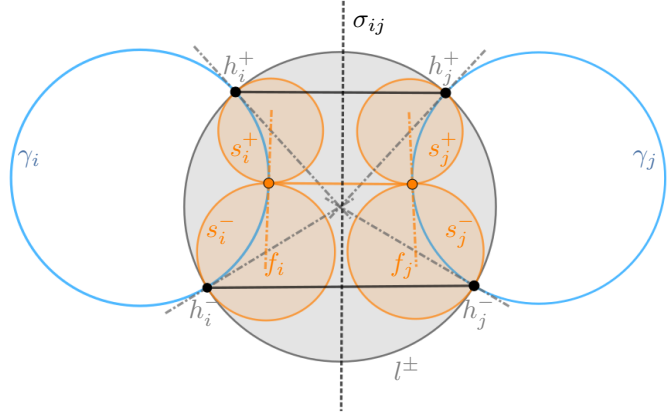

\subsection{Discrete orthogonal systems with Dupin cyclides
 as coordinate surfaces}\label{subsect_dupin_cyclides}
Here we investigate discrete cyclic circle congruences that
stem from a discrete Ribaucour pair consisting of a discrete
Dupin cyclide and a totally umbilic surface. We shall prove
that these special circle congruences yield discrete cyclic
systems where all coordinate surfaces are discrete Dupin
cyclides.

An analogous result in the smooth case can be found in
\cite{darboux_ortho, MR0115134}.

\truepar
Thus, suppose that $f:\mathcal{V} \rightarrow \mathcal{Z}$ is a
discrete Dupin cyclide in the sense of \cite{discrete_channel},
that is, a discrete channel surface with respect to both
coordinate directions.
In particular, any space form projection yields a circular net
with circular curvature lines, so that the corresponding
curvature spheres are constant along them.

Furthermore, we consider a totally umbilic discrete
Ribaucour transform $u$ of $f$ with point spheres on the
constant sphere $n \in \LL$. For genericity, we will assume
that $n$ is not in oriented contact with the Dupin cyclide,
that is, $n\not\in f_i$ for all $i\in\VV$.

The discrete Ribaucour sphere congruence $r$ enveloped
by $(f, u)$ is then provided by the spheres in the contact
elements of $f$ that lie in the linear sphere complex 
$\mathbb{P}(\mathcal{L}\cap n^\perp)$.
Hence, $u$ can be expressed in terms of its (constant)
curvature sphere $n$ and the enveloped Ribaucour sphere
congruence $r$ by $u_i=\lspann{\mathfrak{n},\mathfrak{r}_i}$
(cf \cite{rib_coords, rib_families}).

\truepar
To avoid useless case analyses, a totally umbilic discrete
Legendre map with two families of circular curvature lines
will also be called a discrete Dupin cyclide.

\begin{thm}
Let $f$ be a discrete Dupin cyclide and $u$ a totally umbilic
Ribaucour transform of $f$. Then the orthogonal surfaces
of the discrete cyclic circle congruence associated to the
Ribaucour pair $(f,u)$ are discrete Dupin cyclides.

Furthermore, a suitable choice of contact elements for
the orthogonal surfaces yields a discrete cyclic system
so that all coordinate surfaces are discrete Dupin cyclides.
\end{thm}
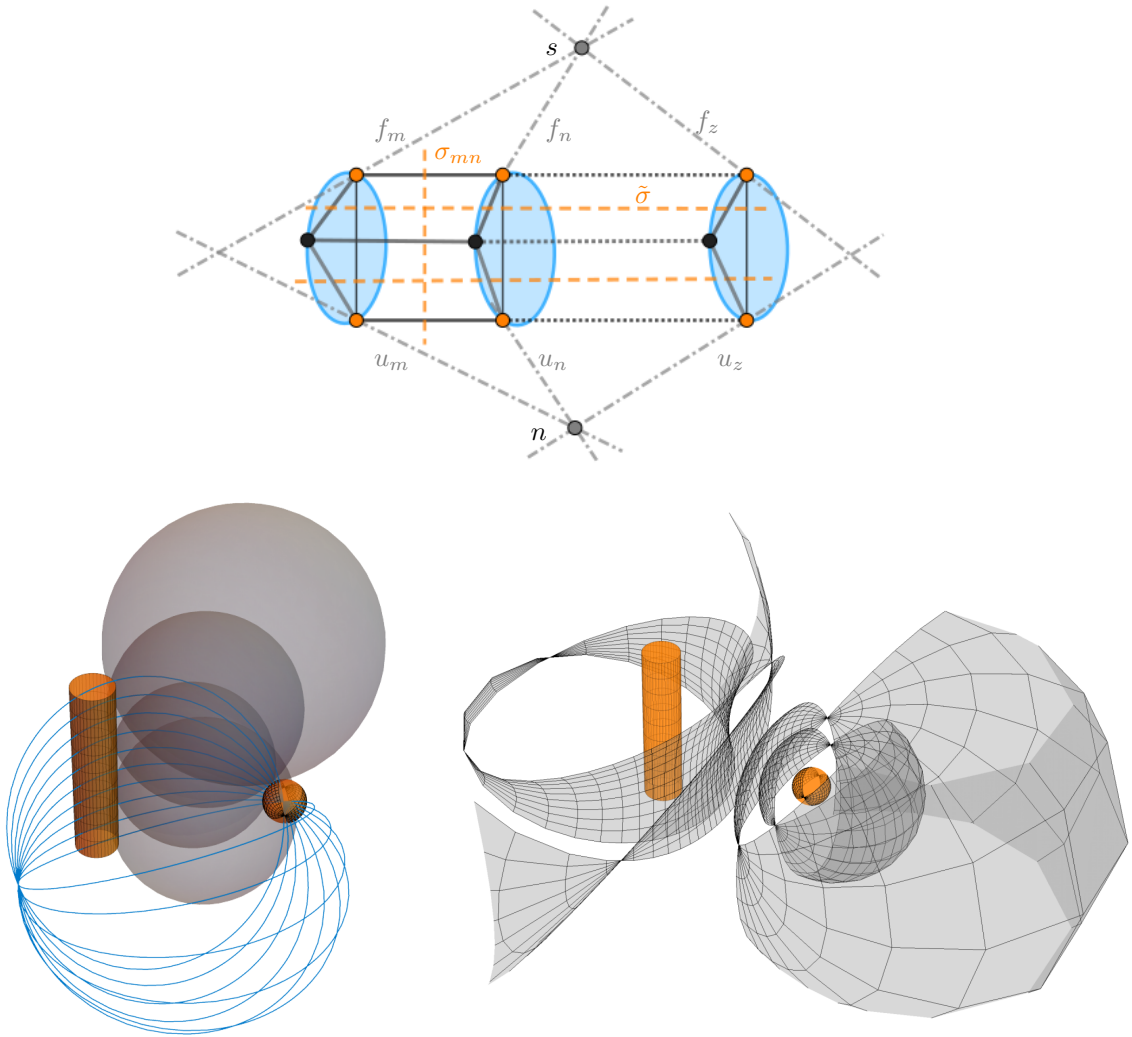
\begin{figure}[t]
\hspace{-3.9cm}\begin{minipage}{6cm}
  \begin{overpic}[scale=.7]{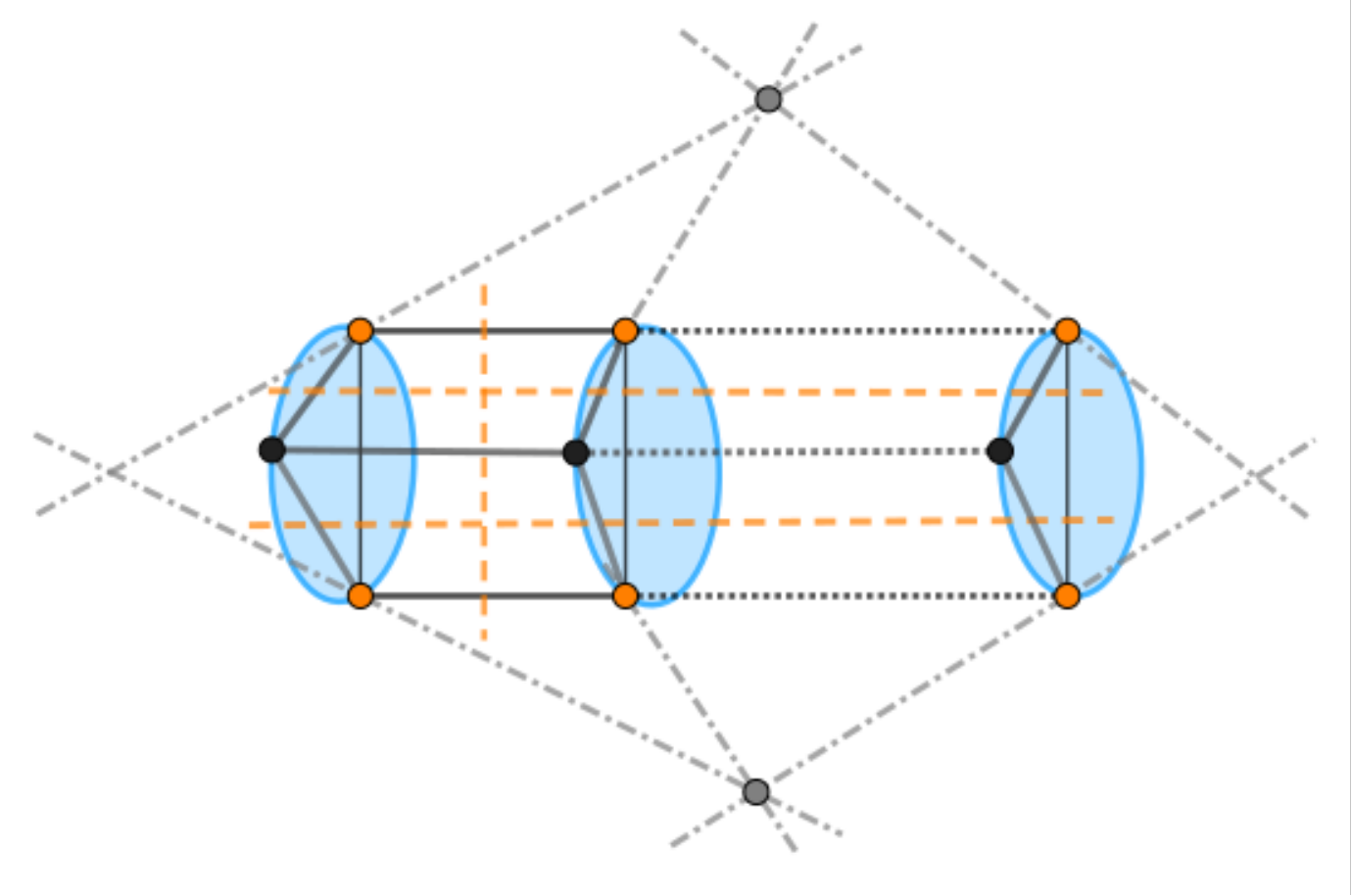}
  	\put(29,16){\color{gray}$u_m$}
    \put(51,16){\color{gray}$u_n$}
    \put(75,16){\color{gray}$u_z$}
    \put(29,47){\color{gray}$f_m$}
    \put(52,47){\color{gray}$f_n$}
    \put(72,48){\color{gray}$f_z$}
    \put(50,6){\color{black}$n$}
    \put(52,58){\color{black}$s$}
    \put(37.1,44){\color{orange}{$\sigma_{mn}$}}
    \put(64,38){\color{orange}{$\tilde{\sigma}$}}
  \end{overpic}
 \end{minipage}
 \\
 \begin{minipage}{5cm}
\hspace*{-2cm}\includegraphics[scale=0.35]{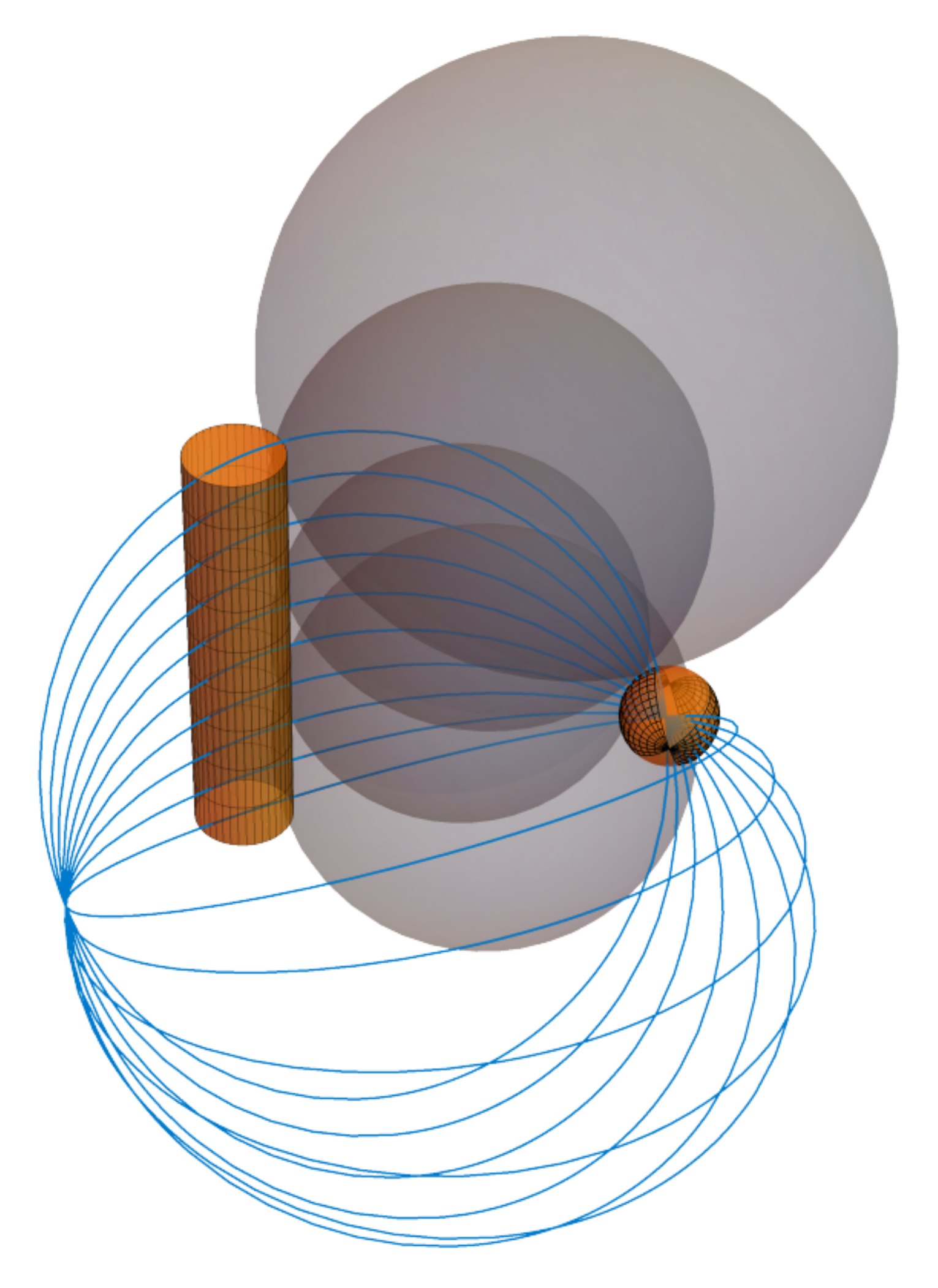}
 \end{minipage}
 \begin{minipage}{0.1cm} \ \ \ \end{minipage}
 \begin{minipage}{5cm}
\hspace*{-1.5cm}\includegraphics[scale=0.28]{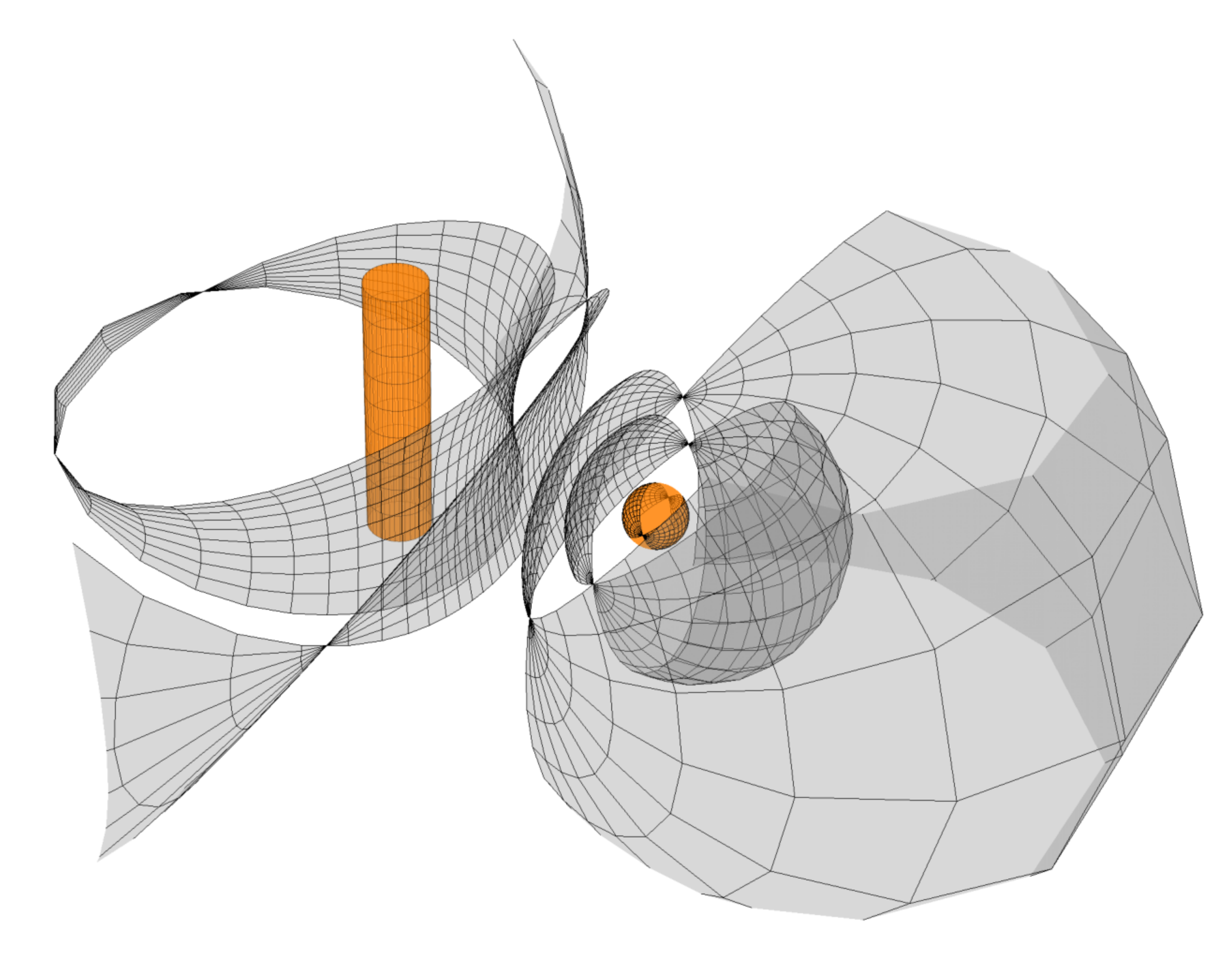}
\end{minipage}
 \caption{A discrete triply orthogonal system associated
  to the Ribaucour pair of a discrete cylinder and a totally
  umbilic spherical transform (orange). Equipped with suitable
  contact elements, all coordinate surfaces of the associated
  discrete cyclic system become discrete Dupin cyclides.}
 \label{fig:dupin_system}
\end{figure}
\begin{proof}
Let $(f,u): \VV \rightarrow \mathcal{Z} \times \mathcal{Z}$
be a discrete Ribaucour pair as in the assumption.
For a proof we pursue the following line of arguments:
\begin{itemize}
\item along each coordinate line of the given Dupin cyclide
 $f$,
 the spheres of the enveloped Ribaucour sphere congruence
 are curvature spheres of a constant Dupin cyclide and,
 therefore, the flat connection for the associated cyclic
 circle congruence is of a special type
 (cf Cor~\ref{cor_rib_flat_Mlie} and 
 Thm~\ref{thm_associated_rib_cyclic});
\item any underlying point sphere map of a discrete
 cyclic system, obtained from a sampling of an
 initial circle, has a special property; namely, its
 ``vertical coordinate surfaces'' are multi-circular nets
 \cite{multinet};
\item hence the M-Lie inversions that relate adjacent contact
 elements of the associated cyclic systems are constant along
 each coordinate ribbon of these
 ``vertical coordinate surfaces'';
\item in particular, propagation of the contact elements of
 the Dupin cyclide $f$ preserves circularity of the curvature
 lines, as well as the fact that the corresponding curvature
 spheres are constant along them;
 \item thus any orthogonal
 surface is a discrete Dupin cyclide and all orthogonal
 trajectories of the point sphere map of the associated
 discrete cyclic system are concircular;
\item moreover, a suitable choice for the contact elements
 of the vertical coordinate surfaces guarantees that those
 are also discrete Dupin cyclides in the Lie sphere 
 geometric sense, namely, discrete Dupin
 cyclides orthogonal to the Dupin cyclides formed by the
 Ribaucour spheres.
\end{itemize}

\truepar
To begin with, we investigate the Ribaucour pair $(f,u)$
and its Ribaucour sphere congruence along a fixed coordinate
line of $\mathcal{V}$. Since the spheres of the contact
elements of $f$ along
each coordinate line all lie in a $3$-dimensional projective
subspace of $\mathbb{P}(\mathbb{R}^{4,2})$,
the spheres of the enveloped Ribaucour congruence along
each coordinate line are curvature spheres of another (constant) Dupin cyclide
(see also \cite{rib_families}). 

Since $f$ is
a discrete Dupin cyclide, along this coordinate line,
all contact elements $F:=\{f_m, f_n, \cdots, f_z \}$
share a common curvature sphere; we denote this sphere
by $s$.
Furthermore,
all contact elements $U:=\{u_m, u_n, \cdots, u_z \}$ intersect
in the constant sphere $n$
(for a schematic see Figure \ref{fig:dupin_system} \emph{top}).
Thus, these two families of contact elements provide
two curvature lines on the Dupin cyclide obtained by
the spheres $R:=\{r_m, r_n, \cdots, r_z\}$ of the enveloped
Ribaucour sphere congruence along this coordinate line.

Therefore, additionally to the M-Lie inversions that
provide the flat connection for the Ribaucour pair (cf Cor
\ref{cor_rib_flat_Mlie}), we obtain further M-Lie inversions:
let $(f_m, u_m)$ and $(f_z, u_z)$ be two
arbitrary pairs of contact elements, then the four 
corresponding point spheres are concircular. Hence, 
there exists an M-Lie inversion $\sigma_{mz}$ so that
\begin{equation}\label{equ_mlie_multi}
 \sigma_{mz}(f_z)=f_m, \ \sigma_{mz}(u_z)=u_m, \ 
  (\text{and } \sigma_{mz}(r_z)=r_m).
\end{equation}
Thus, this M-Lie inversion $\sigma_{mz}$ also exchanges
the circles $\Gamma_m$ and $\Gamma_z$ of the orthogonal cyclic
circle congruence $\Gamma$ associated to the Ribaucour pair 
(see Thm~\ref{thm_associated_rib_cyclic});
hence, also the
point spheres of its other orthogonal surfaces.

\truepar
Next we investigate the underlying point sphere map
of a ``vertical coordinate surface'' of the associated
cyclic system:
using the above M-Lie inversions this is given by a
sampling of an initial circle, say $\Gamma_m$.
We aim to see that this point sphere map is
multi-circular in the sense of \cite{multinet},
that is,
every coordinate quadrilateral is circular,
not just every elementary coordinate quadrilateral.

We recall from \cite{multinet, rib_families}
that multi-circular point sphere nets may be characterized
by the existence of M-Lie inversions that interchange
corresponding point spheres of any two coordinate
lines in one family,
as those of (\ref{equ_mlie_multi}) do.
Thus we obtain multi-circularity of the point sphere net;
and, by symmetry, similar M-Lie inversions
$\tilde{\sigma}_{ij}$ that interchange the point spheres
of coordinate lines in the other family
(as illustrated in Fig \ref{fig:dupin_system} \emph{top}).

\truepar
For corresponding coordinate lines of two adjacent orthogonal
surfaces of the cyclic circle congruence,
the (constant) M-Lie inversion $\tilde{\sigma}$ that arises
from the multi-circularity of the vertical surface may be
used to transport contact elements:
it clearly maps a curvature sphere that is constant along
the coordinate line on one orthogonal surface to an alike
curvature sphere of the other.
Hence all orthogonal surfaces are, with $f$, Dupin cyclides.

Furthermore, we learn that the point spheres along any
coordinate line of the underlying point sphere map of the
discrete cyclic system are circular.

\truepar
Finally, to obtain a discrete cyclic system with discrete
Dupin cyclides as coordinate surfaces,
it remains to equip the vertical coordinate surfaces of the
 underlying point sphere map with suitable contact elements,
that is, to complement each contact element of an orthogonal
 surface of the cyclic circle congruence by two mutually
 orthogonal contact elements that are tangent to the
 corresponding circle of the congruence.

To do so, we fix one initial contact element $f_0$ of the
Dupin cyclide $f$ and consider the two (circular) coordinate
lines of the underlying point sphere map of $f$ that pass 
through the point sphere $p_0\in f_0$:
each is contained in precisely one
(unoriented, M\"obius geometric) sphere
that is tangent to the circle $\Gamma_0$
or, equivalently,
that intersects the Dupin cyclide $f$ orthogonally
along the given curvature line.
By construction, these two spheres are orthogonal,
and choosing an orientation for each of them yields
suitable contact elements.

Further observe that the sphere constructed from one
(circular) curvature line of $f$ is invariant under
the corresponding M-Lie sphere transformations $\sigma_{mn}$
along that curvature line.
Consequently, propagating the just constructed contact
elements at $p_0$ yields a discrete cyclic system that
consists of Dupin cyclides in the Lie geometric sense.
\end{proof}

\noindent We remark that this construction can be generalized
to Ribaucour pairs of two Dupin cyclides, that also lead to
associated cyclic systems with Dupin cyclides as coordinate
surfaces.
Note that this yields a ``totally cyclic system'',
where each coordinate direction provides a cyclic circle
congruence.
Details regarding this construction in the smooth case
can be found in \cite{totallycyclic}.
A suitable sampling then yields a construction
in the discrete case.

\truepar \
\bibliography{mybib}

\begin{thebibliography}{10}

\bibitem{MR239516}
F.~Backes.
\newblock Les syst\`emes hypercycliques et les surfaces de {G}uichard.
\newblock {\em Acad. Roy. Belg. Bull. Cl. Sci. (5)}, 54:219--231, 1968.

\bibitem{blaschke}
W.~{Blaschke}.
\newblock {\em {Vorlesungen {\"u}ber Differentialgeometrie {III}}}.
\newblock Springer Grundlehren XXIX, Berlin, 1929.

\bibitem{MR2657431}
A.~I. Bobenko, H.~Pottmann, and J.~Wallner.
\newblock A curvature theory for discrete surfaces based on mesh parallelity.
\newblock {\em Math. Ann.}, 348(1):1--24, 2010.

\bibitem{org_principles}
A.~I. {Bobenko} and Yu.~B. {Suris}.
\newblock {On organizing principles of discrete differential geometry. Geometry
  of spheres.}
\newblock {\em {Russ. Math. Surv.}}, 62(1):1--43, 2007.

\bibitem{orth_nets_clifford}
A.I. Bobenko and U.~Hertrich-Jeromin.
\newblock Orthogonal nets and clifford algebras.
\newblock {\em T\^ohoku Math Publ}, 20:7--22, 2001.

\bibitem{multinet}
A.I. {Bobenko}, H.~{Pottmann}, and T.~{R{\"o}rig}.
\newblock {Multi-Nets. Classification of discrete and smooth surfaces with
  characteristic properties on arbitrary parameter rectangles}.
\newblock {\em Discrete Comput Geom}, 63:624--655, 2020.

\bibitem{ddg_book}
A.I. {Bobenko} and Yu.~B. {Suris}.
\newblock {Discrete differential geometry: Integrable structure}.
\newblock {\em {Graduate Studies in Mathematics 98, Amer. Math. Soc.,
  Providence}}, 2008.

\bibitem{bcjpr}
F.~E. Burstall, J.~Cho, U.~Hertrich-Jeromin, M.~Pember, and W.~Rossman.
\newblock {Discrete $\Omega$-nets and Guichard nets}.
\newblock {\em arxiv.org/abs/2008.01447}, 2020.

\bibitem{csg_45}
F.E. {Burstall} and D.~{Calderbank}.
\newblock {Conformal submanifold geometry IV-V}.
\newblock {\em {manuscript}}, 2016.

\bibitem{rib_coords}
F.E. Burstall, U.~Hertrich-Jeromin, and M.~Lara~Miro.
\newblock Ribaucour coordinates.
\newblock {\em Beitr. Algebra Geom.}, 60(1):39--55, 2019.

\bibitem{MR2652493}
F.E. Burstall, U.~Hertrich-Jeromin, and W.~Rossman.
\newblock Lie geometry of flat fronts in hyperbolic space.
\newblock {\em C. R. Math. Acad. Sci. Paris}, 348(11-12):661--664, 2010.

\bibitem{lin_weingarten}
F.E. Burstall, U.~Hertrich-Jeromin, and W.~Rossman.
\newblock Discrete linear {W}eingarten surfaces.
\newblock {\em Nagoya Math. J.}, 231:55--88, 2018.

\bibitem{disc_cmc}
F.E. Burstall, U.~Hertrich-Jeromin, W.~Rossman, and S.~Santos.
\newblock Discrete surfaces of constant mean curvature.
\newblock {\em RIMS Ky\^ok\^uroku Bessatsu}, 1880:133--179, 2014.

\bibitem{curvedflats}
F.E. Burstall and M.~Pember.
\newblock Lie applicable surfaces and curved flats.
\newblock {\em arxiv.org/abs/2007.11947}, 2020.

\bibitem{book_cecil}
T.~{Cecil}.
\newblock {\em {Lie sphere geometry. With applications to submanifolds.}}
\newblock Springer, New York, 2008.

\bibitem{orth_cds}
J.~Cie\'li\'nski, A.~Doliwa, and P.M. Santini.
\newblock The integrable discrete analogues of orthogonal coordinate systems
  are multi-dimensional circular lattices.
\newblock {\em Physics Letters A}, 235(5):480--488, 1997.

\bibitem{coolidge}
J.~L. Coolidge.
\newblock {\em A treatise on the circle and the sphere}.
\newblock Oxford University Press, 1916.

\bibitem{darboux_ortho}
G.~{Darboux}.
\newblock {Le\c cons sur les syst\`emes orthogonaux et les coordonn\'ees
  curvilignes.}
\newblock {\em {Paris: Gauthier-Villars.}}, 1910.

\bibitem{disc_flat_front}
J.~Dubois, U.~Hertrich-Jeromin, and G.~Szewieczek.
\newblock Discrete flat fronts in hyperbolic space.
\newblock {\em in preparation}.

\bibitem{MR0115134}
L.~P. Eisenhart.
\newblock {\em A treatise on the differential geometry of curves and surfaces}.
\newblock Dover Publications, Inc., New York, 1960.

\bibitem{MR1752778}
J.~A. G\'{a}lvez, A.~Mart\'{\i}nez, and F.~Mil\'{a}n.
\newblock Flat surfaces in the hyperbolic {$3$}-space.
\newblock {\em Math. Ann.}, 316(3):419--435, 2000.

\bibitem{MR1421285}
U.~Hertrich-Jeromin.
\newblock On conformally flat hypersurfaces, curved flats and cyclic systems.
\newblock {\em Manuscripta Math.}, 91(4):455--466, 1996.

\bibitem{discrete_channel}
U.~Hertrich-Jeromin, W.~Rossman, and G.~Szewieczek.
\newblock Discrete channel surfaces.
\newblock {\em Math. Z.}, 294(1-2):747--767, 2020.

\bibitem{cyclic_guichard_eth}
U.~{Hertrich-Jeromin}, E.~{Tjaden}, and M.~{Z\"urcher}.
\newblock {On Guichard's nets and cyclic systems}.
\newblock {\em arxiv.org/abs/dg-ga/9704003}, 1997.

\bibitem{MR2946924}
T.~Hoffmann, W.~Rossman, T.~Sasaki, and M.~Yoshida.
\newblock {Discrete flat surfaces and linear {W}eingarten surfaces in
  hyperbolic 3-space}.
\newblock {\em Trans. Amer. Math. Soc.}, 364(11):5605--5644, 2012.

\bibitem{MR2196639}
M.~Kokubu, W.~Rossman, K.~Saji, M.~Umehara, and K.~Yamada.
\newblock Singularities of flat fronts in hyperbolic space.
\newblock {\em Pacific J. Math.}, 221(2):303--351, 2005.

\bibitem{L1872}
S.~Lie.
\newblock Ueber {{Complexe}}, insbesondere {{Linien}}- und
  {{Kugel}}-{{Complexe}}, mit {{Anwendung}} auf die {{Theorie}} partieller
  {{Differential}}-{{Gleichungen}}.
\newblock {\em Math. Ann.}, 5(1):145 -- 208, 1872.

\bibitem{MR947546}
P.~Moon and D.~E. Spencer.
\newblock {\em Field theory handbook. Including coordinate systems,
  differential equations and their solutions.}
\newblock Springer-Verlag, Berlin, second edition, 1988.

\bibitem{rib_cyclic}
A.~Ribaucour.
\newblock M\'{e}moire sur la th\'{e}orie g\'{e}n\'{e}rale des surfaces courbes.
\newblock {\em Journal de math\'{e}matique pures et appliqu\'{e}es}, 4:5--108,
  219 -- 270, 1891.

\bibitem{rib_families}
T.~{R\"orig} and G.~{Szewieczek}.
\newblock {The Ribaucour families of discrete R-congruences}.
\newblock {\em {Geometriae dedicata}},
  https://doi.org/10.1007/s10711-021-00614-1, 2021.

\bibitem{disc_sing}
W.~Rossman and M.~Yasumoto.
\newblock {Discrete linear Weingarten surfaces with singularities in Riemannian
  and Lorentzian spaceforms}.
\newblock {\em Advanced Studies in Pure Mathematics}, 78:383--410, 2018.

\bibitem{salkowski}
E.~{Salkowski}.
\newblock {Dreifach orthogonale Fl\"achensysteme}.
\newblock {\em {Encykl. d. math. Wiss. III D 9, 541-606}}, 1921.

\bibitem{totallycyclic}
G.~Szewieczek.
\newblock Totally cyclic systems.
\newblock {\em in preparation}.

\end{thebibliography}

\truepar \ 

\truepar
\begin{minipage}{7cm}
 \textbf{Udo Hertrich-Jeromin} \\
 \textbf{Gudrun Szewieczek}
 \\ TU Wien, 
 \\ Wiedner Hauptstra\ss e 8-10/104,
 \\ 1040 Vienna, Austria
 \\ udo.hertrich-jeromin@geometrie.tuwien.ac.at
 \\ gudrun.szewieczek@geometrie.tuwien.ac.at
\end{minipage}
\end{document}